\documentclass[12pt]{amsart}
\usepackage{fancyhdr,latexsym,amssymb,amsmath,amsthm}
\usepackage{comment,graphicx,color,listings}
\usepackage{morefloats}
\usepackage[section]{placeins}
\definecolor{red1}{rgb}{1,0.9,0.9}
\definecolor{blue1}{rgb}{0.9,0.9,1}
\definecolor{green1}{rgb}{0.9,1,0.9}
\definecolor{yellow1}{rgb}{1,1,0.9}
\definecolor{yellow2}{rgb}{1,1,0.8}
\newtheorem{thm}{Theorem}
\newtheorem{coro}[thm]{Corollary}
\newtheorem{conj}[thm]{Conjecture}
\newtheorem{lemma}[thm]{Lemma}
\newtheorem{propo}[thm]{Proposition}
\theoremstyle{definition}

\let\paragraph\subsection

\title{Eulerian edge refinements, geodesics, billiards and sphere coloring}
\author{Oliver Knill}
\date{August 21, 2018}
\address{
        Department of Mathematics \\
        Harvard University \\
        Cambridge, MA, 02138
        }

\begin{document}
\maketitle

\begin{abstract}
A finite simple graph is called a 2-graph if all of its unit spheres $S(x)$ are cyclic 
graphs of length $4$ or larger. A 2-graph $G$ is Eulerian if all vertex degrees of $G$ are even. 
An edge refinement of a graph adds a new vertex $c$, replaces an edge $(a,b)$ by two edges $(a,c),(c,b)$ 
and connects the newly added vertex $c$ with the vertices $u,v$ in $S(a) \cap S(b)$. We 
prove here two theorems. Theorem I assures
that every $2$-graph can be rendered Eulerian by successive edge refinements. 
The construction is explicit using geodesic cutting. After the refinement, we have
an Eulerian 2-graph that carries a natural geodesic flow. We construct some ergodic ones. 
A 2-graph with boundary is finite simple graph for which every unit sphere is either 
a path graph $P_n$ with $n \geq 3$ vertices or a cyclic graph $C_n$ with $n \geq 4$ vertices.
2-balls are special 2-graphs are simply connected with a circular boundary.
Theorem II tells that every 2-ball can be edge refined using interior edges
to become Eulerian if and only if its boundary has length 
divisible by 3. Also this is constructive. 
A billiard map is defined already if all interior vertices
have even degree. We will construct some ergodic billiards in 2-balls,
where the geodesics bouncing off at the boundary symmetrically and 
which visit every interior edge exactly once. A consequence of Theorem II is that 
an Eulerian billiard which is ergodic must have a boundary length that is divisible by $3$. 
We also construct other 2-graphs like tori with ergodic geodesic flows.
This clashes with experience in the continuum, where tori have periodic points
minimizing the length in homology classes of paths. 
Ergodic Eulerian 2-graphs or billiards are exciting because they satisfy a 
Hopf-Rynov result: there exists a geodesic connection between any two vertices. We get so
unique canonical metric associated to any ergodic Eulerian graph. It is
non-local in the sense that two adjacent vertices can have large distance.
\end{abstract}

\section{The results}

\paragraph{}
A finite simple graph for which every unit sphere is a cyclic graph with four
or more vertices is a discrete {\bf two-dimensional manifold}. We call such a discrete
surface a {\bf 2-graph}. Examples are 4-connected triangulations of a 2-manifold.
The Euclidean realization of the simplicial Whitney complex of such a graph
is a piecewise linear 2-manifold homeomorphic to a smooth compact 
2-dimensional manifold so that the topology is well understood. 
If we allow unit spheres also to include path graphs, then the set of vertices $v$
for which $S(x)$ is a path graph form the {\bf boundary} and is a finite union of cycles.
A {\bf $2$-graph $G$ with boundary} is then a {\bf discrete 2-manifold with boundary} and
the geometric realization of the Whitney complex of such a graph $G$ is homeomorphic 
to a smooth $2$-manifold with boundary. 

\paragraph{}
Odd degree vertices in a surface are interesting because these vertices are obstructions for 
$G$ being Eulerian and in the simply connected case for $G$ having {\bf chromatic number} $3$. 
By the Euler-Hierholzer theorem, graphs for
which all vertex degrees are even are the same than graphs which feature an Eulerian circuit, 
a closed path visiting every edge in $G$ exactly once. An other motivation is
the possibility to define a {\bf geodesic flow} or {\bf billiard dynamics}: if a vertex 
is even, then there is a natural way to continue an ``incoming ray" to propagate to an 
``outgoing ray". For discrete manifolds with boundary, we need both for coloring reasons 
as well as for billiard dynamics purposes only the interior 
vertices to have even degree as if a path hits an odd degree boundary vertex straight on, 
we just revert the path. We can now ask under which conditions it is possible that a 
$2$-graph with boundary can be edge refined using interior edges to become Eulerian. 

\paragraph{}
Our first result applies to all discrete surfaces without boundary:

\begin{thm}[Theorem I]
\label{1}
For every 2-graph there is a sequence of edge refinements which renders the graph Eulerian.
\end{thm}

\paragraph{}
The proof is constructive and the edge refinement can be done in polynomial time. 
We call the proof the {\bf ``geodesic self healing algorithm"} as one can just let a geodesics
run and if reaching an odd degree vertex point, let ``the geodesic do the edge cutting".

For discrete surfaces with boundary, we have the next theorem. We say that an edge $e=(a,b)$ is 
an {\bf interior edge}, if not both vertices $a,b$ are in the boundary of $G$. An interior edge
can however hit the boundary in one vertex. 

\begin{thm}[Theorem II]
\label{2}
Let $G$ be any 2-ball. There exists a sequence of edge refinement steps using interior edges of $G$
which renders the graph Eulerian if and only if the boundary circle of $G$ has a length 
which is divisible by $3$. 
\end{thm}

\paragraph{}
Also this proof is constructive. The necessity of the length being a multiple of $3$
is easy to see but for the sufficiency, we have to provide a cutting algorithm which 
renders the graph 3-colorable with a periodic coloring of the boundary. Let us in the
next paragraph give an argument which convinces why Theorem II is true. It is a 
discrete mean curvature flow procedure which takes a ball and deforms it to a point. 
It is not the proof we give here however. 

\paragraph{}
First cut up edges reaching the boundary if necessary in order that all edge degrees of
the boundary points are even. The  mod 3 condition assures that we can color the boundary periodically
with three colors. The evenness of the degree at the boundary points assures that the 
induced coloring of the triangles attached to the boundary lead to no coloring conflict on the
other side. This gives the inner boundary of these triangles a 3-periodic coloring so that
also the inner boundary has a length which is a multiple of 3. We can  continue like that
taking care of the next layer. This reduction works until different parts of the boundary 
start to collide. There are "singularities" of this ``wave front dynamics" similarly as with 
caustics in differential geometry as different wave points come together. But we can using
double edge refinements refine a graph locally at places where we start to get close to slow
down the wave in these areas. This``curvature flow" can be continued with enough padding 
near parts of the boundary with negative or zero curvature. If guided well, the wave hits 
the final stage of a wheel graph with boundary length divisible by $3$. 
Since this argument needs to be fleshed out and since we did not 
implement this ``discrete mean curvature flow" it certainly tells why the result is true. 

\paragraph{}
For the proof given here, we therefore use an other idea and use gather the odd 
degree vertices together, annihilate until we have only 2 left close to each other in a common wheel. Then we can use a local
procedure to fix and annihilate the last two. This strategy has been implemented in a computer
and works pretty fast. 

\paragraph{}
So far, we only have covered the simply connected case of a 2-ball. 
There are examples of annuli with two boundaries, where both boundaries have length not divisible by 
3 but where nevertheless, we can edge refine. An example is the discrete annulus
$(K_2 \times C_4)_1$. The periodic 3 condition therefore does not generalize to more general topologies 
and there might be no simple counting conditions for the boundary in general which are equivalent 
to be able to edge refine the graph. 

\section{The Euler puzzle}

\paragraph{}
The problem to find a sequence of edge refinement steps which renders a graph Eulerian 
is a {\bf puzzle}. We call it the {\bf Euler puzzle}. It can be played with any finite simple graph. 
A {\bf move} in this game is an edge refinement step: take an edge $e=(a,b)$, place
a new vertex $c$ in the middle and connect it to all points in the intersection 
$S(a) \cap S(b)$ of the unit spheres of $a$ and $b$. 
The goal is to apply a sequence of edge refinement steps to end up with
an Eulerian graph. For which graphs can we win the game? 
For a one-dimensional graph, the game is only winnable if it is already
in a winning state from the beginning: the reason is that any given vertex keeps the vertex 
cardinality fixed. In this paper we show that for two-dimensional graphs, the game can always be won
if there is no boundary and that in the case of the boundary, the length of the boundary matters. 

\paragraph{}
The game can become challenging for larger graphs, even with the strategy given.
When playing the game on paper, this solitary game is ``like trying to solve a Rubik's 
cube that is fighting back"\footnote{quoting ``Q", played by Ben Whishaw in Bond movie ``Skyfall"},
the reason being that every refinement move increases the number of vertices and
edges in the graph and so potentially makes it harder for the player to win. We certainly felt
like this when playing the game, even if assisted by the computer. 
The interest origins of course in the 4-color theorem
as in 3 dimensions ``winning of the game" means ``coloring a planar graph". Playing the
game in a 3-graph locally is like fixing faces of the Rubik cube and requires fixing spheres and
2-disks. But while playing the repeated cutting complicates the game more and more for
future cuttings. The theorems covered here cover both cases. 

\paragraph{}
Theorem I be reformulated by saying:

\begin{coro}
The Euler puzzle is winnable for any 2-graph without boundary.
\end{coro}

\paragraph{}
The game can also always be won for surfaces $M$ with boundary 
if one allows cutting the boundary edges. But as we are interested in coloring
the boundary $\delta M$ (especially in higher dimensions) \cite{knillgraphcoloring},
we {\bf prefer not to cut the boundary edges} and keep a {\bf boundary condition}. 
In that case, the puzzle has a {\bf boundary constraint} which is formulated
in the following corollary: 

\begin{coro}
The Euler puzzle on a $2$-ball is winnable with fixed boundary
if and only the length of the boundary is divisible by $3$. 
\end{coro}

\paragraph{}
We can move the odd degree vertices to the boundary by playing the game but
can not necessarily remove all the odd degree vertices on the boundary. 
If the length of the boundary is divisible by $3$ this is possible. As the
vertex degree $d(x)$ of boundary points $x$ is now in the set $\{ 4,6,8, \dots \}$, the
boundary curvature $1-d/3$ is negative at every point. This means that
for a refined graph, there are lots of interior points, where the curvature $1-d/6$ is
positive.  The procedure ``Edge sub-division moves negative curvature to the 
boundary and positive curvature into the interior." 

\paragraph{}
If we want to bring in an allegory from relativity, where curvature 
is associated with mass, the edge refinement ``generates mass away from the boundary"
as the boundary radiates mass away.  (This would be nice to be understood also in higher dimensions.
For three dimensional graphs, the edge refinement seems to transport  discrete ``Ricci type curvature"
which is a quantity for edges $(a,b)$ and given by $R(e) = 1-d/6$, where $d$ is the cardinality
of $S(a) \cap S(b)$. The analogue Eulerian condition in three dimensions is that $d$ is even
at every point. )

\paragraph{}
Edge refinement is differential geometrically different from Barycentric subdivision, as the
later which keeps the curvature balance essentially constant in neighborhoods. 
Whether one should to see the mass transport from the boundary to the interior this as a 
manifestation of some cosmological principles remains to be seen, as it is pure speculation. 
It is certainly mathematically interesting. 
But it is important that this phenomenon in two dimensions only occurs if there is some boundary. 
Without boundary, the total curvature of the surface, the sum $\sum_x 1-d(x)/6$ is 
constant by discrete Gauss-Bonnet. 

\section{Particles}

\paragraph{}
A more topological aspect comes in if we think of the graph as physical space and a 
vertex with {\bf odd degree} as a {\bf particle}, which manifests in some sort of 
{\bf defect} or {\bf anomaly}. It is an interesting phenomenon that
there are particle pairs in a disc which can not be annihilated within the disc. 
In some sense, there are {\bf two type of particles} and only particles of the same type can 
be removed simultaneously. This is related to the fact that we can not
realize a 5-7 degree configuration on an otherwise flat torus (\cite{IKRSS,Izmestiev2013}),
the reason being that the two particles are of different type and can not annihilate
each other. Let us reformulate the result in the following way:

\begin{coro}
Given a simple closed path on a 2-sphere of length which is not a multiple of 3. Assume that
for all vertices on that curve, the vertex degree within the full graph is even. 
If there is a particle on one side of the curve (and so a particle pair) 
then there must be a particle (and so a particle pair) the other side of the curve.
\end{coro} 

\paragraph{}
Such phenomena were observed by Jendrol and Jucovic in 1972 \cite{JendrolJucovic72} 
(see also \cite{Jendrol1975}, where modulo 3 conditions appear but it seems unrelated of what
we do here) and relate to a theorem
of Eberhard from 1891 \cite{Eberhard1891} which covers the realization of curvature configurations of convex
polyhedra. Nearly regular polyhedra were looked at also in \cite{Crowe1969} who reports that the first
investigations of this type appeared in the 1967 edition of \cite{gruenbaum}.  
A result close to what we do here has been formulated by Fisk \cite{Fisk1978} in a pointed way: 
{\it if there are only two particles on a sphere, they can not be adjacent.}
For more recent extensions of these principles, see \cite{Izmestiev2015}. 
We can show something similar:

\begin{coro} 
Given a $2$-ball $G$ with boundary of length $3k$, it is not possible 
to have two adjacent odd degree vertices on the boundary 
if no other odd degree vertices are present in $G$.
\end{coro}

\paragraph{}
The reason is that if there was, we could add an other triangle, removing
the oddness of these two vertices and have a boundary length $3k+1$ which is
not divisible by $3$. This contradicts Theorem II. From this follows the 
just mentioned Fisk observation because we could cut from such a sphere
a disk with boundary length $3k$ which has the two vertices as boundary points. 

\paragraph{}
We see that particles come in two flavors and only particles of 
the same type can be combined and destroyed by edge refinement. 
A configuration without particles is then the {\bf vacuum}. Because of the Euler handshake 
formula, the number of particles is even at any time also when restricted to a subgraph. 
By edge refinement we can create or destroy pairs of particles. The question is whether we 
can get rid of all the particles by pair destroying them, possibly after creating other particles 
first. The answer is yes, on any closed 2-graph, but it is no in general, if we have a boundary
as we need the mod 3 condition. 

\paragraph{}
In \cite{KnillEulerian}, on page 21, we state in a lemma that any 2-sphere
can be edge refined to become Eulerian. Theorem I pushes this a bit further as it 
generalizes it from 2-spheres to arbitrary 2-graphs. This had not been clear to us
before as we would have expected that for a different topology, some obstruction might 
develop.  We actually see in Theorem II that there can indeed be obstructions if 
the graph has a boundary. The reason for the interest in this problem 
was that if we ask the same question for a 2-disk with boundary, 
then we can only render the graph Eulerian through interior edge refinements,
if the boundary length is divisible by $3$. This key insight can also be used
to show that every 2-graph with boundary can be refined to be particle-free. 

\begin{coro}
If a disc has a boundary length not divisible by $3$ and no particle are on the boundary,
then there must exist particles in the interior of the disk. 
\end{coro}

\paragraph{}
It is interesting that a metric property of the boundary relates with topological 
and (as we will just see also differential geometric)
consequences in the interior of the region. It does not surprise as one can see the
length condition as a cohomology condition for a $Z_3$-valued cohomology which needs to be
trivial. The number $3$ is since this is the minimal coloring number of a 2-dimensional sphere.
The existence of odd degree vertices is also of 
topological nature as we will see next that it has relations with geodesic flow
or more generally with billiards, if boundaries are present. But also the odd degree
property relates to a $Z_2$-valued cohomology, and again this relates to the fact that 2
is the minimal coloring number of a 1-dimensional sphere. 
We don't understand the general conditions yet which have to be met to render a graph fully
Eulerian. Neither if multiple boundaries are present nor do we know yet what 
happens in higher dimensions. 

\section{Geodesic flow} 

\paragraph{}
The property of being Eulerian on a 2-graph $G$ is equivalent of having a 
natural {\bf geodesic flow} on the graph $G$. Let us explain. Obviously, if we draw a connection from 
one vertex $x$ to a neighboring vertex $y$ and want to continue the path naturally, we need
the vertex degree of $y$ to be even. There is no such natural geodesic flow on the icosahedron graph
for example because we would have to specify a rule at every vertex how to continue: do we
take the left or right choice? We could flip a coin and so get some sort of random walk.

\paragraph{}
As a side remark, as pointed out in earlier papers, we
can also use the unitary Schr\"odinger flow $e^{i D t} u$ solving $u_t=i Du$ with $D=d+d^*$ on the exterior
bundle leading to the wave equation $u_{tt}=-D^2 u=Lu$. These are ordinary differential 
equations as we have a finite simple graph. The paths are geodesics in a finite dimensional
unitary group. They remain geodesics (but with a different metric) also when doing non-linear 
deformation of $D$ using a Lax pair. 

\paragraph{}
The just mentioned quantum approach solves the
Hopf-Rynov problem as it replaces the discrete geodesic ``hopping" with a nice path in a unitary
group. Indeed, one can find a complex initial condition such that the wave is at time $t=0$
at some vertex and at some later time $t=T$ at an other vertex. The problem is that for most other
times, the wave then will  be non-localized and have support on different vertices of the finite graph.
If we want a more ``classical geodesic flow", we need the graph to be Eulerian. 
We will see that edge refinement essentially can be done by ``self healing" along geodesic paths. 

\paragraph{}
Having an Eulerian graph gives us a nice geodesic {\bf dynamical system} but no Hopf-Rynov:
given two vertices, there is in general no geodesic
connecting them. We can only restore Hopf-Rynov if we have an {\bf ergodic geodesic
flow} meaning that there is one geodesic trajectory (which always is periodic due to the finiteness of the
graph) covers all edges and so is an {\bf Eulerian path}. It is a rather strange Hopf-Rynov because
we might have to travel for a long time to connect two points which are close in this {\bf ergodic geodesic metric}
on the graph. We see this also if the geodesic path is used to ``pair annihilate" two particles. The two
particles need to be in the same connected component of the ``geodesic metric". 

\paragraph{}
Let us reformulate the proof strategy of Theorem I. We will explain it in 
more detail in the next section.

\begin{coro}
Letting the geodesic flow cut through the graph produces
an Eulerian graph on which a global geodesic flow exists.
\end{coro}

\section{Proofs}

\paragraph{}
The proof of Theorem I is done by ``letting the geodesic flow fix itself":

\begin{proof} 
Start with a $2$-graph $G$. Take a vertex $x$ and a ``unit direction", meaning
to take an edge $(x,y)$ attached to $x$. Now, if the edge degree of the neighboring
vertex $y$ is even, we can just continue the flow with the opposite edge $(y,z)$, where 
$z$ is in the antipode of $x$ on the sphere $S(y)$. Do the same at $z$ 
if $z$ is an even degree vertex.  Continue the flow until we reach
a vertex with odd vertex degree. If that happens at a point $v$, and we
came from $u$, we have two antipodes $a,b$ to $u$ on $S(v)$. Make an
edge refinement $(a,b) \to (a,w,b)$ of the edge connecting these two 
vertices $a,b$ and continue the flow with $u,v,w$. The flow has fixed the vertex $v$.
After the refinement, the vertex $v$ has now even degree. We might have
destroyed the evenness of a perfectly even degree vertex $w$
but if that is the case, we just cut through that vertex again and fix so $w$
back. Continue like that, cutting, until the geodesic path $\gamma$ hits an edge which 
has been traversed before. Now, all the vertices along the path have even 
degree and especially, every newly added point has even degree. Start
with an other vertex not yet covered and continue the ``geodesic healing".
The only way that the procedure would go wrong is to start at a vertex $x$ with
odd degree and have the geodesic path hitting only even degree vertices until
returning back. But if this happens in every direction, we have a pairing
of the edges of $S(x)$ so that $x$ would have even degree. 
\end{proof}

\paragraph{}
When implementing on a computer, we start at an odd degree vertex $x$ and start immediately
cutting fixing $x$ and continue until we reach an other odd degree vertex $y$. Now
the pair $x,y$ has been annihilated. Now search for the next odd degree vertex and 
start cutting until reaching a second one. This pair destruction works because if there
would be an odd vertex $x$ left, which can not be connected to an other odd degree vertex,
then we have that pairing argument of the edges in the unit sphere $S(x)$ and so a contradiction
to the oddness. 
The task of rending the graph Eulerian is of polynomial complexity. 
A rough upper  bound is a quadratic complexity in the number of edges $E$
the reason being that in order to fix a vertex $x$, we might have to cut $E$
edges. Since by the Euler polyhedron formula, $|V|+|F|=|E|+\chi(G)$, a
quadratic complexity $O(|E|)^2$ suffices.

\paragraph{}
Here is a proof of Theorem II. It is still a bit clumsy. We feel there should be a
much easier argument similarly to the geodesic cutting procedure.
We are given a disc $G$ with
boundary: we want to show that a 2-ball $G$ can be rendered
Eulerian if and only if the boundary length $|\delta G|$ is divisible by 3.
The idea is to remove (pair annihilate) all except two odd degree particles, then gather them 
into a local situation and realize that the metric boundary condition  $|\delta G|=3k$
produces constraints on how the points can be located in a small disk, then check that
for a small half wheel graphs with that metric boundary condition, we can do the
cutting. 

\begin{proof}
(i) The necessity of the divisibility condition is easier to see:
start with an Eulerian disk $G$. It is 3-colorable by Kempe-Heawood.
Because by assumption, all vertex degrees at the boundary are even,
the coloring obtained from fixing the colors on a first triangle forces the colors
of adjacent triangles. The coloring is therefore 3-periodic on the boundary. 
This implies that the boundary length is a multiple of 3. \\
(ii) To show the converse, we have to verify that
if the boundary length is a multiple of 3, it is possible to refine the graph using
edge refinements of interior edges to become Eulerian. There are a two basic construction steps 
which allow us to move or merge odd degree vertices. If the graph should be too narrow at
some point for the procedure to apply, we simply pre-apply some {\bf double edge cuttings} (two
cuttings at the same vertex), which refines the graph without changing the vertex parity. \\

A) {\bf (Switching even and odd pairs)}: Given an edge $e=(a,b)$ where the vertex
$a$ has odd degree and the vertex $b$ has even degree. We can perform 
several cuts along a closed loop in a wheel graph centered at $b$, then cut from 
$a$ to $b$. This renders $a$ even and $b$ odd.
With this procedure, we can move all odd degree vertices into a common wheel graph
or half wheel graph. \\

B) {\bf (Reducing from 3 to 1)}: Given an odd degree center $c$ and two odd degree vertices $a,b \in S(c)$. 
We can make a sequence of cuts in the disc $D(c)$ which renders both $a,b$ even. 
This procedure allows us to remove two in a triplet of odd degree vertices.
An other important ingredient is that since the sum over all vertex degrees is 
$2|E|$ by the {\bf Euler handshake formula}, there is always an even number of odd degree vertices. \\

Using these two procedures A) and B), we can reduce the number of odd degree
vertices until there are none or exactly $2$. In the former case we are done. 
In the later case, we can move them all to the boundary contained in a common disk. 
Lets do that and call these odd degree vertices $P$ and $Q$.  \\

We can assume that there is an interior or boundary point $R$ such that $P,Q$ are in 
a half wheel graph centered at $R$ (allowing $P,Q$ to be at the center of the half-wheel). 

We claim that the two points $P,Q$ at the boundary can not be adjacent. If they were, we could
glue a new triangle to $G$ at the edge $PQ$ and get a graph $G^+$ with boundary length $3k+1$. 
As adding the triangle has changed the degree at $P$ and $Q$ by one, the two odd degree vertices
have disappeared leading to an Eulerian 2-ball with boundary not divisible
by $3$. In part (i) of this proof we have already shown that this is not possible. 

We need to show that we can remove the two non-adjacent $P,Q$ at the boundary. 
This is now a local case. 

(Here is a special case when $G$ has no interior points. The disc $G$ is then a shellable 
``tree polytop" obtained by attaching triangles in a tree like manner without containing any wheel graph. 
In that case one can by induction show that the divisibility condition allows a cutting. 
The induction assumption consists of two cases, the half wheel consisting of three triangles and having
boundary 6 and the graph with 4 triangles building a equilateral triangle of boundary length 6. The
induction assumption is to remove either a half wheel or such a ``triangle" ). \\

(Remark: We can assume that the boundary length is larger than 3 as otherwise, filling a triangular cell
gives a sphere. In this sphere, we can find a geodesic cut (as in the proof of Theorem I 
which connects $P$ with $Q$ and there is only one direction either $P \to Q$ or $Q \to P$ which 
cuts through the triangle. So, we can repair the two points without involving the triangle. 
We can therefore assume that the boundary length is at least $6$. ) \\

We have now reduced the situation to a graph with two odd non-adjacent boundary vertices $P,Q$ 
contained in half wheel graph $H$ centered at some vertex $R$ (which might agree with 
$P$ or $Q$ or not). \\

If the graph $G$ is is equal to $H$, then $H$ must have a boundary length which is divisible by three so 
that it is a sequence of edge refinements using edges in $H$ that renders the vertex degrees of $P,Q$ even. 
Doing that for a wheel graph or half wheel graph with boundary length $3k$ can be directly done.

If the boundary length of $H$ is $3k-1$, we add an other triangle first to it (this is possible if
$G$ is not equal to $H$) and again use the local edge refinement of a ``wheel graph or half wheel with an ear" 
to render the graph Eulerian.  If the length of $H$ should be $3k+1$, we add first two more triangles 
to get a graph with boundary length divisible by $3$, where we can make the reduction. 

So, the only cases to be checked to be reducible are therefore half wheel graphs with boundary, 
a half wheel graph with an additional triangle attached or a
half wheel graph with two more additional triangles attached. In all these 
three cases, we can give explicit local edge refinements using edges in $H$ only.
These local edge refinements also work within $G$.
\end{proof}

\paragraph{}
Also this procedure is implemented in a computer. It is in general a bit harder if the graph has
narrow parts, where the boundaries are connected. A possibility to avoid this is to
first make {\bf double edge refinements}, meaning to replace an edge $(a,b)$ inside with 
a path $(a,u,v,b)$ and connect both $u$ and $v$ to $S(a) \cap S(b)$. Unlike edge refinements,
the double edge refinements do not change the parity of the points in $S(a) \cap S(b)$. But
it can be used to ``build up some interior". By the way, also this double edge
refinement can pump curvature from the boundary to the interior and so builds up some
``dark matter" inside the disk, rendering the boundary carry more and more negative
curvature. As double edge refinement pair produces or annihilates particles, one can associate
this transfer of mass with boundary radiation as known in cosmology. 

\section{Ergodic Eulerian graphs} 

\paragraph{}
In the continuum, there are always periodic orbits of the geodesic flow in a non-trivial 
homology class of a surface. The idea is simple. Take a torus for example and pull a string
along a non-contractible path on the surface. Now shorten the string until it is no more possible.
The limiting case is a geodesic as it locally minimizes the distance. 
These periodic orbits do not necessarily exist in the
discrete. The reason is the observation in the following corollary.

\begin{coro}
There are ergodic discrete Eulerian 2-graphs $G$ having the topology of a sphere
or a torus. 
\end{coro}

\paragraph{}
These examples are constructed by 
trial and error. We typically need to construct a few dozen examples until we reach an ergodic 
one. We have included the code for determining the number of ergodic components of an 
Eulerian graph. 

\paragraph{}
If there are particles in the form of odd degree vertices present in a graph, then
the geodesic flow or for a billiard needs to be {\bf stopped} at those particles as there is
no natural way to continue. We still can used this dynamical system to define a 
{\bf metric} between two points. It is the length of the minimal geodesics
between them. The distance is assumed to be infinite if there is no connection, meaning
that the graph is disconnected. A consequence of the inability of removing particles is:

\begin{coro}
A 2-ball with a boundary of length not divisible by 3 is always disconnected
in the billiard metric. 
\end{coro}

\paragraph{}
If we have edge-refined the graph to have only 2 particles $x,y$ left (which is the
best we can do), then we can not find a geodesic from  $x$ to $y$. They live in 
two ``parallel universes". 

\paragraph{}
Let us mention that if $G$ is any 2-graph, then its 
Barycentric refinement $G_1$ is always Eulerian (the coloring is provided by 
the dimension functional). The vertices of $G_1$ are the simplices in $G$ 
and two are connected if one is contained in the other. 
The Barycentric refinement of a 2-graph is always Eulerian. Can they become
ergodic? The answer is yes: 

\begin{propo}
There are ergodic Barycentric refined 2-spheres,
Barycentric refined 2-tori or projective planes.
\end{propo}


\section{Billiards}

\paragraph{}
If no odd degree vertices exist in the interior of a 2-graph
with boundary, a billiard dynamical system
is defined. While on a 2-graph without boundary or Eulerian graph with 
boundary, all geodesic orbits are closed and
no edge is traversed twice, there is an other possibility for boundary edges
with odd degree. A geodesic path starting perpendicular to the boundary (meaning
that the ray $(x,y)$ hits the boundary point has the property that $x$ in $S(x)$
has equal distance to the boundary points of $S(x)$) will have
to hit the boundary perpendicularly a second time at an other boundary of the graph
and so go cover each edge twice forwards and backwards. We call such a 
billiard path {\bf undirected} as its path can not be assigned an orientation. 
Billiards can be defined therefore if {\bf interior vertices} have even degree. 

\begin{coro}
In a 2-ball with boundary for which every interior point has even degree and a boundary 
length is not divisible by 3, there is at least one undirected (self backtracking) billiard orbit. 
Such a billiard table can not be ergodic. 
\end{coro}

\paragraph{}
More precisely, the number of undirected billiard paths is equal to half of the 
number of odd degree vertices at the boundary. But we will see that there
are discrete analogues of Bunimovich stadium \cite{ChernovMarkarian,Tab95} in 
the sense that the billiard is ergodic. Of course, there is no hyperbolicity as 
the orbit of this dynamical system is just a cyclic path. 

\begin{coro}
There are ergodic Bunimovich type discrete billiards in the shape of 2-balls.
\end{coro}

\paragraph{}
We can not use a geodesic cutting procedure (where billiards replace the geodesic flow)
for proving Theorem II. The reason is that we might
reach the boundary in an edge which would force us to cut the boundary
which is not allowed. We have not yet got conditions which assure that we
can render a general 2-graph with boundary Eulerian. If we glue
together 2n triangles along a linear cycle
leading to an annulus with two boundary curves, both with length $n$, we always
get Eulerian situations as there are no interior points and all boundary vertices
have degree $4$.  The mod 3 condition obviously does not matter in such an annulus. 
The reason is that unlike the disk, the annulus is not simply 
connected any more, so that minimal coloring question are no more decided 
with local conditions. 

\section{The setup}

\paragraph{}
We make the definitions a bit more general as our eventual goal is to understand things 
in higher dimensions too and especially in dimension $d=3$, where we hope that a relatively
simple argument leads to a {\bf constructive 4-coloring procedure} for planar graphs and more generally
will establish that any $d$-sphere has chromatic number $d+1$ or $d+2$ for which no proof does
exist yet. 

\paragraph{}
{\bf Definition.} A {\bf $d$-sphere} $G$ is a finite simple graph which is either the empty
graph $0$ (the case for $d=-1$) or which has the property that for which every unit
sphere $S(x)$ is a $(d-1)$-sphere and such that there exists a vertex $x$ such that 
the graph $G-x$ without vertex $x$ is contractible. A graph is {\bf contractible} if it is either 
the $1$-point graph $1=K_1$ or if there exists $x$ such that $G-x$ and $S(x)$ are both contractible. 
A graph is a {\bf $d$-ball} if it is $G-x$ where $G$ is some $d$-sphere. The {\bf boundary} of a 
ball is generated the vertices $x$ for which $S(x)$ is a $(d-1)$-ball and the interior 
is generated by the vertices $x$ for which $S(x)$ is a $(d-1)$-sphere. The boundary
of a $d$-ball by definition is a $(d-1)$-sphere.  \\

\paragraph{}
Because the 4-color theorem is equivalent to the statement that 2-spheres
have chromatic number $3$ or $4$ (the class of 2-spheres is the class of 4-connected maximally 
planar graphs and coloring the later allows to color every planar graph), 
it is natural to look at the higher dimensional case \cite{KnillEulerian} and to explore 
the conjecture: 

\begin{conj}[Sphere coloring conjecture]
Every $d$-sphere has chromatic number $d+1$ or $d+2$. 
\end{conj}

\paragraph{}
The statement is obvious for $d=1$. It is the {\bf $4$-color theorem} for $d=2$. 
In order to produce a new proof of the $4$-color theorem, one has to embed a given 2-sphere $G$ into a 
3-sphere $H$ and make edge refinements in $H$ without 
refining edges in $G$, until $H$ is Eulerian. Then $H$ is $4$-colorable and the coloring induces
a coloring of the sub graph $G$. The idea to do that is very simple: 
start with a suspension $H$ of $G$ and then edge refine the 3-sphere without using 
edges in $G$ until $H$ is Eulerian. 
As we only have to color one side of $H$, it is enough to start with
a cone extension $H=G+x$ which is a 3-ball, then edge refine until we have an Eulerian 3-ball.
This can be achieved by cleaning out unit sphere after unit sphere inside the ball. 
For now we just state that it is in this part, where we need the edge refining result proven here. 

\paragraph{}
Despite the fact that the coloring of a $1$-sphere (a cyclic graph) with $3$ colors is trivial, 
the case $d=1$ already illustrates the proof strategy in general: 
let the 1-sphere be the boundary of a 2-ball $G$. Now edge refine this
disk $G$ until the interior is Eulerian. (We don't need to have the boundary Eulerian too). 
A coloring of a single triangle then defines the color everywhere so that we have a coloring
of the entire graph including the boundary. It is already here, where an other question 
appears: under which conditions it is possible to color the disc so that also the vertices 
at the boundary have even degree? 

\paragraph{}
We saw that the divisibility by 3 matters. Analogue questions can be asked in higher dimension. 
We should stress however that in order to get the $4$-color theorem, one only has to be able to edge 
refine the 3-ball $H$ using edges from the interior.  
This will then $4$-color the boundary, the $2$-sphere. In order to edge refine $H$ we need to be sure
that we can edge refine 2-spheres and that is what we do here. 

\section{Dual spheres}

\paragraph{}
The following definition of duality appeared already in \cite{KnillEulerian}. 
We present it here as a duality between spheres, identifying a simplex $x$ with its boundary 
sphere $\hat{x}$. This is justified as it is possible to
associate a complete graph with its boundary sphere which is a simplicial complex 
(even so it is no more the Whitney complex of a graph). 

\paragraph{}
{\bf Definition} A complete subgraph $K_{k+1}$ of the graph $G$ is a 
{\bf $k$-simplex} of $G$. It defines a {\bf simplicial complex}, the $(k-1)$-skeleton 
complex which has as a Barycentric refinement a $(k-1)$ sphere. 
Given a d-graph $G$ and a $k \leq d$ simplex $x=(x_0, \dots, x_{k+1})$ in $G$, the
{\bf dual sphere of $x$} is the $(d+1-k)$-sphere $\hat{x} = S(x_0) \cap S(x_1) \cap \cdots \cap S(x_k)$.
We think of $x,\hat{x}$ as a pair of {\bf dual spheres} because if $y_1, \dots, y_m$ are the vertices
of $\hat{x}$, then $x = S(y_0) \cap S(y_1) \cap \cdots \cap S(y_k)$.  \\

\paragraph{}
{\bf Remark.} For a general sub-graph $A$ of a graph $G$, the 
{\bf dual graph} $\hat{A}=\bigcap_{v \in V(A)} S(v)$ is usually
empty, but for the $(k-1)$-complex defined by a $k$-simplex we have {\bf duality} $\hat{\hat{x}}=x$. 
If $G$ is a simplicial complex, the set 
$\{ \hat{x} \; | \; x \in G\}$ is again a set of sets. It is a {\bf co-complex} however,
a set of sets closed under the operation of taking {\bf sup-sets} different from $V$. 
The collection of all $(d-1)$-spheres $x+\hat{x}$ (obtained by joining $x + \hat{x}$)
encodes the simplicial complex, as we can get back the simplices $x$ from the prime factorization 
in the Zykov monoid of spheres. When seen from this angle, the geometry of simplicial complexes is a 
``geometry of spheres". \\

\paragraph{}
{\bf Examples:} The dual sphere of a vertex (=$0$-simplex) $x$ is the 
{\bf unit sphere} $S(x)$ of $x$. The dual sphere of a $(d-1)$ simplex $x$ in a 
$d$-graph is a $0$-sphere; the two points represent then the orthogonal
complement of $x$. The dual sphere of a $d$-simplex $x$ is the empty graph, the $(-1)$-sphere. 
The dual sphere of a $(d-2)$-simplex is a circular graph. It is this situation 
which is the most important for us. 

\section{Eulerian spheres}

\paragraph{}
{\bf Definition.} The {\bf degree} of a $(d-2)$-simplex $x=(x_1, \dots x_{d-1})$ in a $d$-graph $G$ is
the length of the circular graph $\hat{x}$. For $d=2$, it is the 
{\bf vertex degree} of $x$ and for $d=3$, it is the length of the circle $S(a) \cap S(b)$ if $x=(a,b)$.
In the case $d=4$, it is the length the circle $S(a) \cap S(b) \cap S(c)$ if $x=(a,b,c)$. \\

\paragraph{}
{\bf Definition.} A $d$-sphere is called {\bf Eulerian}, if all its $(d-2)$-degrees are even.  
The Euler-Hierholzer theorem assures that for $d=2$, this is equivalent to the existence 
of an Eulerian path as the $(d-2)$-degree is then the vertex degree. 
A generalization of a result of {\bf Kempe-Heawood} in the case $d=2$ is: 

\begin{propo}[Kempe-Heawood generalization]
An Eulerian $d$-sphere has chromatic number $d+1$. 
\end{propo}

\begin{proof} 
For $d=1$, this is clear so that we can assume $d>1$.
Start with coloring one $d$-simplex with $(d+1)$ colors. This defines the colors of the vertices
of every simplex adjacent to it. Continue like this. The degree condition assures that
the monodromy of the coloring has no constraints when closing a loop in the dual graph, 
meaning to build a closed chain of simplices hinging at a common $(d-2)$-simplex. 
We can continue coloring larger and larger neighborhoods. 
As $S^d$ is simply connected for $d>1$, we never have a compatibility 
problem and can color all $d$-simplices and so all graphs. 
\end{proof}

\paragraph{}
The proof demonstrates why simply connectedness is important: if there is a closed loop
we can form a chain of $d$-simplices winding around the graph which can not be collapsed. 
In order to be able to have compatibility, we need the number of simplices to be even. 
as we then can alternate the $(d+1)$'th color along the circle. 
This result therefore generalizes to simply connected $d$-graphs, graphs in which all unit spheres
are $(d-1)$ spheres and which are simply connected in the obvious way (defined in the discrete
but equivalent to the Euclidean realization being simply connected). 

\paragraph{}
{\bf Definition.} Given a $d$-graph $G$ and an edge $e=(a,b)$, an {\bf edge refinement} is the 
finite simple graph $H$ in which the vertex set is increased by an additional vertex $c$ inside $(a,b)$ and 
where the edge set is augmented by the additional edges $\{ (c,x) \; | \;  x \in S(a) \cap S(b)  \}$. 
Formally $T_{(a,b)}( (V,E) ) =  (V \cup \{c\}, (E \setminus \{(a,b)\}) \cup \{ (a,c),c,a) \} 
  \cup S(a) \cap S(b) \})$. 
This is again a $d$-graph as $S(c)$ is the suspension of the sphere $S(a) \cap S(b)$ 
and inductively, $S_H(y)$ is the edge refinement of $S_G(y)$ if $e$ is in $S_G(y)$.  \\

\paragraph{}
Edge refinement can also be realized as a sequence of homotopy deformation step. But unlike a homotopy,
it preserves also the dimension of $d$-graphs and is therefore a topological notion. A differential
geometric aspect comes in as we use edge refinements to fix space up allowing a geodesic flow. 
Edge refinement preserves the class of $d$-graphs, graphs for which every unit sphere is a $(d-1)$-sphere. The 
inverse is an edge collapse. Edge collapses in general do not preserve the class of $d$-graphs.
But we can call two $d$-graphs {\bf e-homotopic} if they have a common edge refinement. 
A simple but powerful observation used by Fisk which will be used to color a sphere is:

\begin{lemma}
If a sphere $G$ is embedded in a larger dimensional sphere $H$ of chromatic
number $c$, then the chromatic number of $G$ is smaller or equal than $c$.
\end{lemma}
\begin{proof}
If $G_1$ and $G_2$ are two graphs and
if $G_1$ is a sub-graph of $G_2$, then the chromatic number of $G_1$ is smaller
or equal than the chromatic number of $G_2$. 
\end{proof} 

\paragraph{}
The sphere coloring conjecture would follow from:

\begin{conj}
For a $d \geq 2$-sphere $G$ which is a sub-graph of a $(d+1)$-sphere $H$, 
it is possible to edge refine $H$ with edges different from $G$ 
such that the modified $H$ becomes Eulerian.
\end{conj}

\paragraph{}
Because the Eulerian host graph $H$ has chromatic number $d+1$ or $d+2$<
also $G$ has then chromatic number $d+1$ or $d+2$. Because $G$ contains $d$-simplices
which have $d+1$ vertices all connected to each other, the chromatic number can not be smaller 
than $d+1$. \\

\paragraph{}
By the discrete Brouwer-Schoenfliess theorem, the $d$-sphere $G$ in the $d+1$ sphere $H$
divides the later into two 2-balls $H_1$ and $H_2$. We only need to color a
ball $B$ having $G$ as a boundary with $d+2$ colors. This coloring with $d+2$ colors is 
possible if the graph is Eulerian, meaning that if the degree of every $(d-2)$-simplex in 
$B \setminus G$ is even. 

\paragraph{}
Given a $d$-sphere $G$ containing a subgraph $K$ so that $G-K$ is a ball
and the boundary length of this ball is divisible by 3 in the case $d=2$,
we hope then to be able to edge refinements in $G$ without edges in $K$ 
such that every degree of in the interior of $G-K$ becomes even. 

\paragraph{}
In dimension $1$, there is nothing to show as the degree of an interior point is always 2.
In dimension $2$, we have to do edge refinements so that all vertex degrees are even. In this
case, the 3-divisibility condition is needed. 
In dimension $3$, we will have to get all the edge degrees even. We currently do not see
any constraint there but that has to be investigated still. 
Here is a basic simple principle which relates the even index condition of an higher dimensional 
simple with a lower dimensional one: 

\begin{lemma}
Let $x$ be a $(d-2)$-simplex containing the vertex $v$ and $y = x \cap S(v)$. 
Then ${\rm deg}(x) = {\rm deg}(w)$. 
\end{lemma}

\paragraph{}
For example, in the case $d=3$, if $x=(v,w)$ is an edge, build the 2-sphere $S(v)$.
The {\bf edge degree} of $x$ is the {\bf vertex degree} of $w$ in $S(v)$. 
The proof is the definition as we just intersect with a sphere.

\section{Edge refining disks}

\paragraph{}
We have seen in Theorem I that for 2-sphere $S$ there is a sequence of edge 
refinements which renders $S$ Eulerian.
We hope that we can bootstrap from this a similar result in  higher dimensions. 
Given a $3$-sphere $S$, we have seen that we can push the odd degree 
edges away to an embedded 2-sphere. There, we reduce one dimension more. 

\begin{conj}
$d$-spheres have edge refinements which are Eulerian.
\end{conj}

\paragraph{}
{\bf Definition} Given a $d$-sphere $G$, let $H$ be the collection of $(d-2)$-simplices for which the 
degree is odd. We call it the {\bf defect variety} in $G$. In a $2$-sphere this is a discrete set of points.
In a $3$-sphere it forms a one dimensional pure graph without end points after normalization. Let us call a d-dimensional pure
graph {\bf closed} if it has no boundary. A boundary point of $H$ is a point for which $S(x) \cap G$ is contractible.

\begin{lemma}[Defect varieties]
Defect varieties are closed. 
\end{lemma}
\begin{proof}
Let us look at the case $d=3$: we have a collection of edges. Let $x=(a,b)$ be an edge and 
$S(a) \cap S(b)$ the circular graph $\hat{x}$. Look at the sphere $S(b)$. As it has one odd
degree vertex, there must exist an other odd degree vertex as well as a few other pairs. 
\end{proof} 

\section{Remarks}

\paragraph{}
The reverse operation of an {\bf edge refinement}
$R_{(a,b)}:   (V,E) \to (V \cup \{c\},(E \setminus \{(a,b)\}) \cup \{(a,c),(b,c) \} \cup 
     \{ (c,d) \; | \;  d \in S(a) \cap S(b) \})$
is an {\bf edge contraction} $C_{(b,c)}$
in which the edge $(b,c)$ is removed, vertices $b,c$ are merged and edges $(x,b),(x,c)$ are identified.
In the category of simple graphs, any multiple edges are identified so that both edge refinement and
edge contraction preserve the class of finite simple graphs. While we can reverse an
edge refinement with an edge contraction, we can not always reverse an edge contraction
with an edge refinement as edge contraction can change dimension: an edge
contraction of $K_n$ is $K_{n-1}$ but $K_1$ has no edges any more so that $K_2 \to K_1$
can not be reversed. Edge contraction also can change the topology and dimension. A
contraction of the one dimensional $C_4$ gives $K_3$ which is two dimensional.
On the other hand, edge refinement never changes the maximal dimension of the simplicial
complex, nor the nature of the topology.

\paragraph{}
Both edge refinement as well as edge
contraction can be realized as discrete homotopy deformations.
Unlike homotopy deformations,
edge refinements $R$ are also ``continuous" deformations as they preserve dimension.
The graphs $G, R(G)$ are homeomorphic
in the sense of \cite{KnillTopology}, which generalizes the notion of homeomorphism
when seeing graphs as one-dimensional simplicial complexes. We always see a graph equipped
with the Whitney complex, the set of vertex sets of complete subgraphs rather than the
$1$-dimensional skeleton complex $V \cup E$. So, we establish here that the
Eulerian property is no topological obstacle, at least not in two dimensions. In three
and higher dimensions, it might well be different and one reason why graph coloring of
2 spheres is so much harder than graph coloring of 1-spheres.

\paragraph{}
There is not only a homotopy or homeomorphism picture, there is also
a differential geometric aspect as the Eulerian property implies 
the existence of a geodesic flow on the graph. A natural question is under
what conditions we have a flow in higher dimensions. 
Again, in three and higher dimensions, we
need a bit more structure to establish that. It is only in two dimensions, that an
evenness condition establishes the existence of an antipodal map on the unit spheres.
In three dimensions already, we need the unit spheres $S(x)$ on which there is an involutive
graph automorphism without fixed points, allowing a geodesic to propagate through $x$.

\section{Questions}

\begin{itemize}
\item Can we always edge refine a given 2-graph to become Eulerian and additionally 
achieve that the geodesic flow is {\bf ergodic}? The same question can be asked for
billiards in a 2-disk with boundary length divisible by 3. The geodesic flow on an Eulerian 
graph is called {\bf ergodic} if it covers every edge at least ones. This is an Eulerian path.
Apropos: we know that 2-graphs are always Hamiltonian already. But how frequent is 
the set of ergodic Eulerian 2-graphs? 

\item Given an Eulerian 2-graph and 2 vertices $a,b$, the shortest connection between
$a$ to $b$ might not be a geodesic flow in the sense discussed here. In other words, 
there is no {\bf Hopf-Rynov theorem} for finite graphs. The simple reason is that 
there are more vertices in the graph than in the unit sphere. But we can ask:
is it possible for every pair $a,b$ in the graph to make further edge refinements so that
there is a shortest geodesic between a and b? If yes, then we can define from the Eulerian
graph $G$ a new graph $H$ such that every pair of vertices in $G$ has a geodesic in $H$. 
Now, since $G$ is a subgraph of $H$, we get so a sequence $G_n$ of graphs and so a 
pro-finite limit $\overline{G}$, which is the {\bf geodesic completion} of $G$. This is no
more a finite graph but it is a model for a 2-dimensional space in which Hopf-Rynov
works. We believe that the pro-finite limit of the Barycentric refinement sequence $G_n$
also has this Hopf-Rynov property. 

\item We currently believe that the result has higher dimensional versions:
all d-graphs, combinatorial discrete manifolds (finite simple graphs
for which all unit spheres are $(d-1)$ spheres) can be edge refined to become Eulerian. 
We know already that these manifolds are always Hamiltonian. As there are tori with chromatic
number 5, we know also that we can in general not edge refine a discrete 3-manifold with boundary 
using edge refinements in the interior. For {\bf simply connected} 3-manifolds with boundary it
could however to be possible. 

\item Not every graph can be the image of a refinement. The icosahedron is an example,
for which the odd particle density is 1. Assume, we are given a graph and want the opposite 
of Eulerian, have as many odd degree vertices as possible.  What is the 
maximal {\bf odd vertex degree density} we can achieve through edge refinement? 
Can we always get to density 1? 

\item For coloring the boundary, we don't need the even degree condition at the boundary. 
What is the analogue of divisibility by 3 in higher dimensions? Is there a condition at
all in dimensions $d \geq 3$ which corresponds to the divisibility condition 
in two dimensions? Or are there no conditions necessary any more to render
a 3-ball Eulerian. 

\item This result relates to an
observation of Jendrol and Jucovic about impossible triangulations of the 2-torus.
But there is not always a connection. There are Eulerian graphs like 
4-8 pairs which show the phenomenon. For us here, the 4-8 pair is not a particle pair
We are interested in odd degree vertices. 

\item 
It is known that the chromatic number of a $2$-torus can be $5$. The torus can be
seen as part of the solid torus. As there are 2-tori with chromatic number 5, 
the coloring arguments do not go through. One reason is that Kempe-Heawood needs the
space to be simply connected. The Eulerian property of a solid torus does not imply that 
it can be colored with 4 colors. We would like also to understand the
defect structure, which now be a homotopically non-trivial closed curve
\end{itemize}

\section{Illustrations}

\begin{figure}
\scalebox{0.17}{\includegraphics{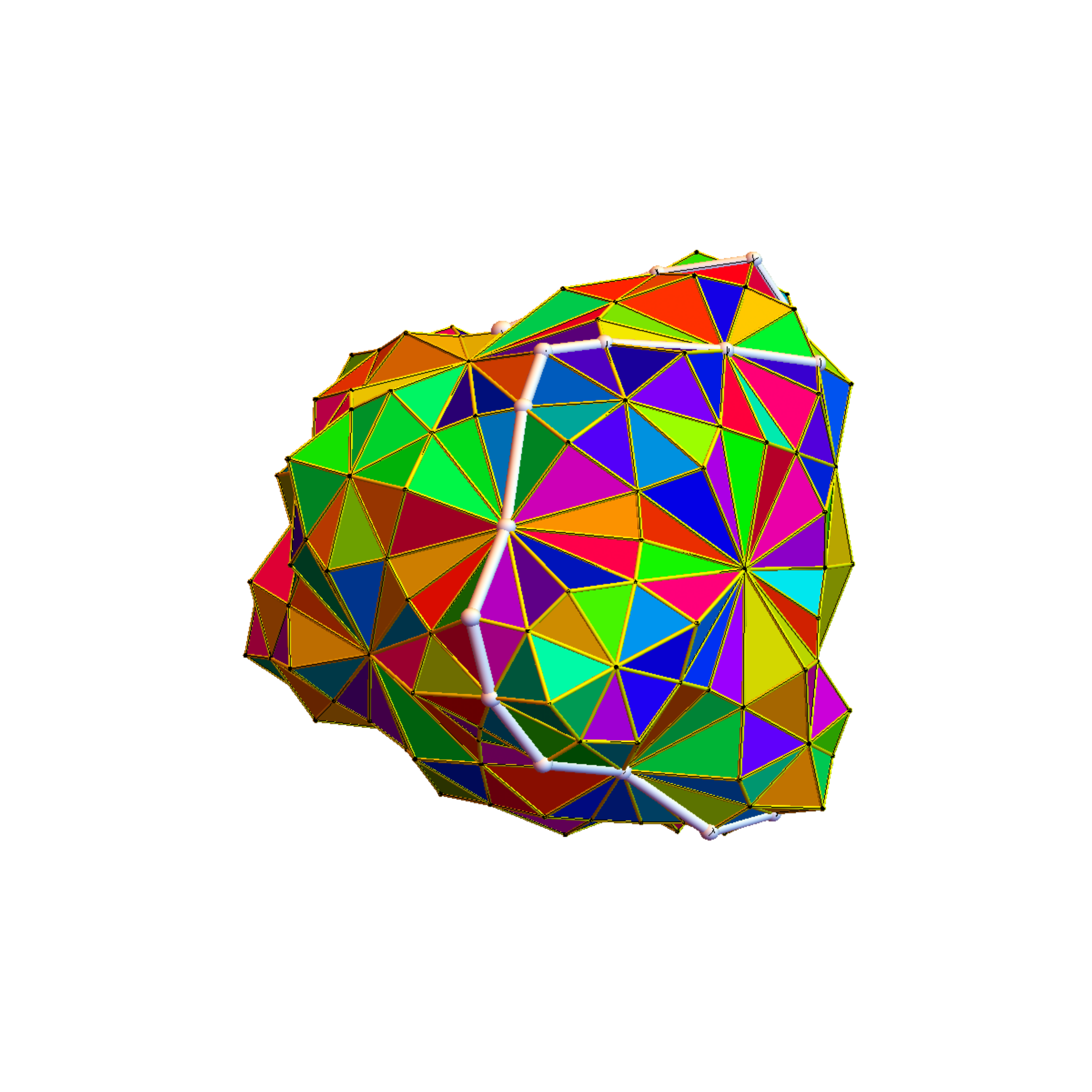}}
\caption{
\label{Barycentric}
For an Eulerian 2-graph the geodesic flow
is defined. This figure shows an ergodic sphere a 2-sphere
for which the flow has has only one component.
The geodesic is then an Eulerian path: a closed path which visits
every edge exactly once.
}
\end{figure}

\paragraph{}
Ergodic spheres can already be obtained by starting with an icosahedron,
running a flow to make the graph Eulerian. We provide in an appendix
code which allows to get Eulerian refined graphs from any 2-graph. 

\begin{figure}
\scalebox{0.17}{\includegraphics{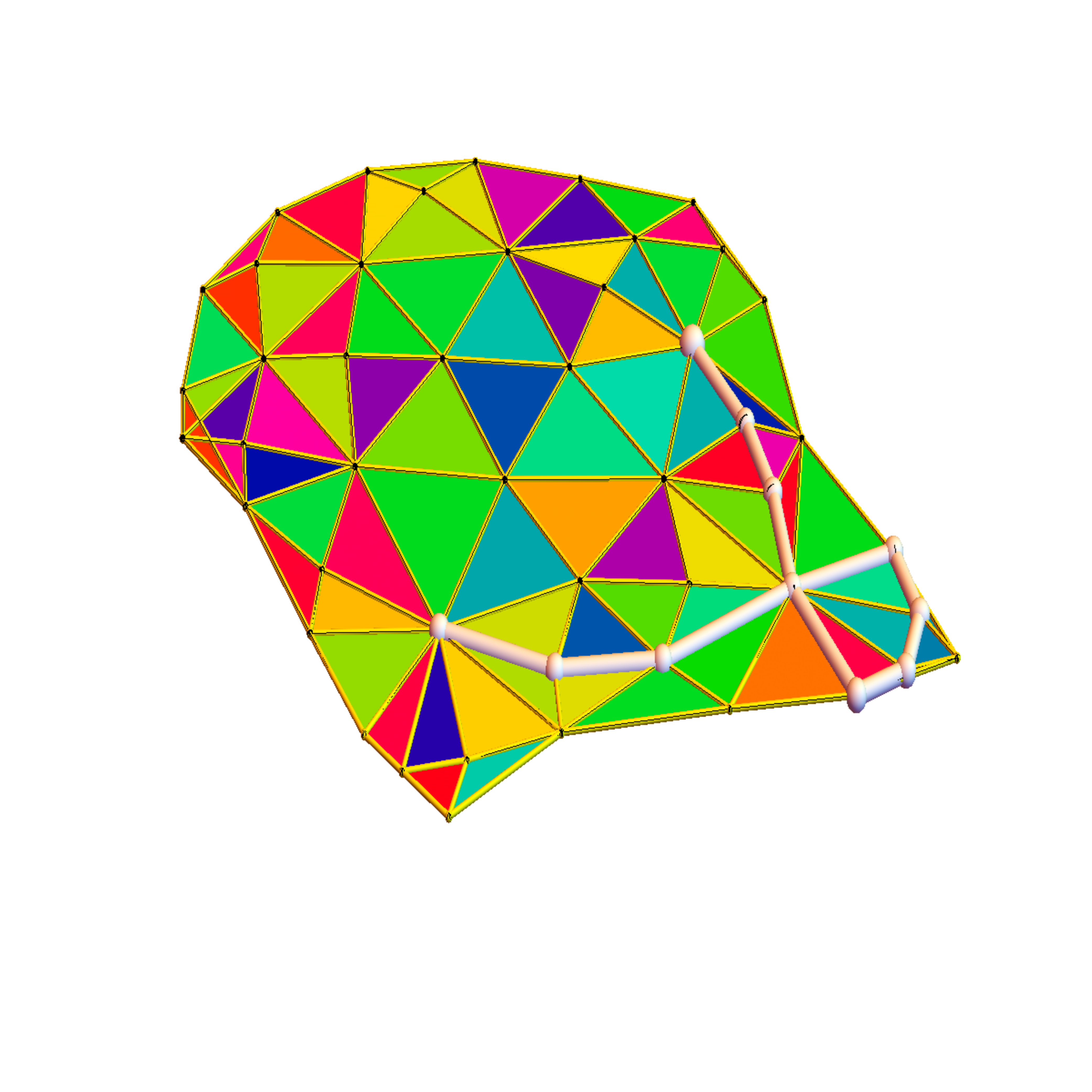}}
\caption{
\label{Billiard}
Part of a billiard trajectory in an Eulerian graph. 
The graph is a 2-ball. It had been refined to become Eulerian. 
}
\end{figure}

\begin{figure}
\scalebox{0.12}{\includegraphics{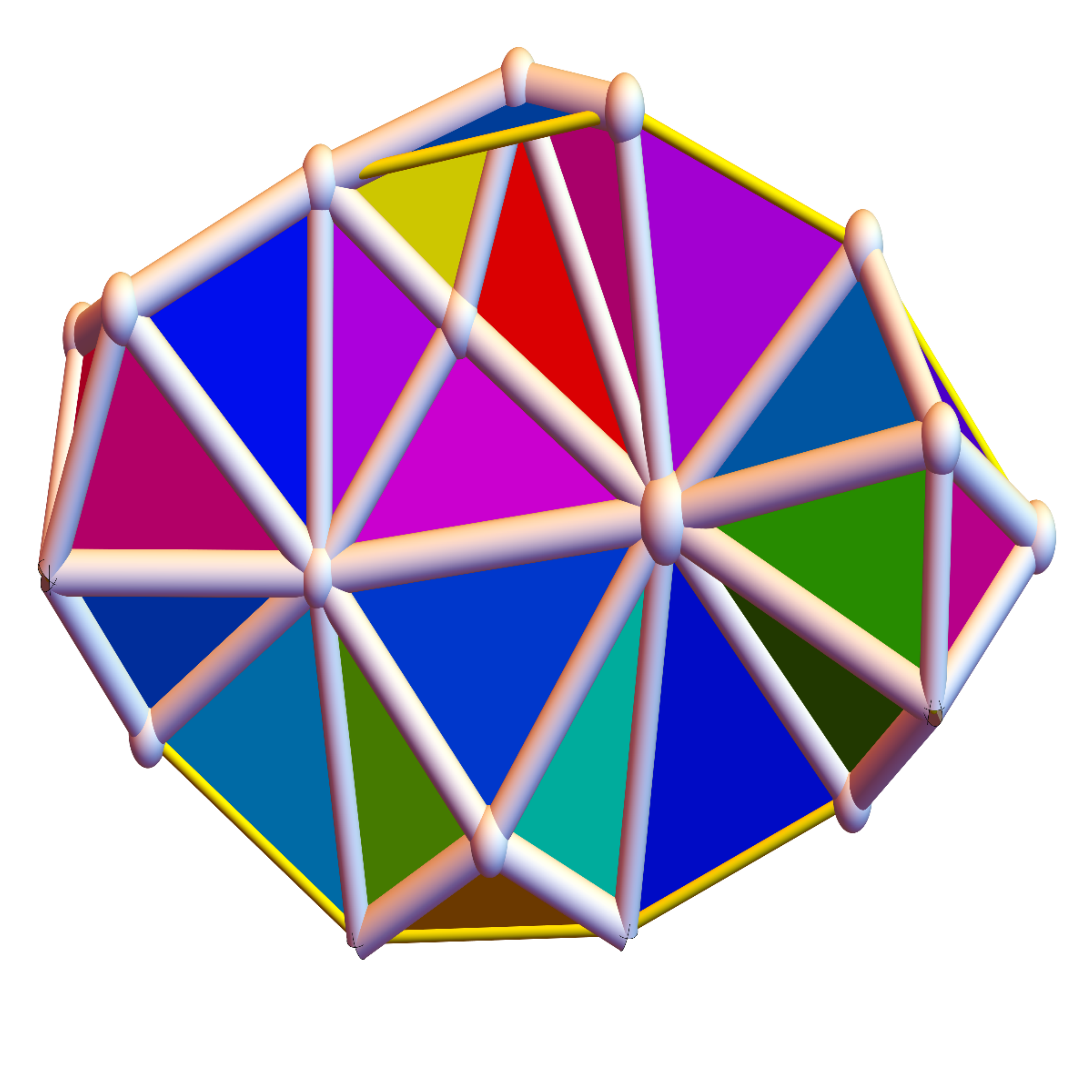}}
\caption{
\label{Bunimovich}
An example of a Bunimovich type discrete billiard. It is ergodic. 
There is also the boundary path but this is not counted as
an ergodic component when looking at discrete billiards.
As we have an ergodic billiard, we have Hopf-Rynov: there is 
a geodesic connection between two points.
}
\end{figure}

\begin{figure}
\scalebox{0.17}{\includegraphics{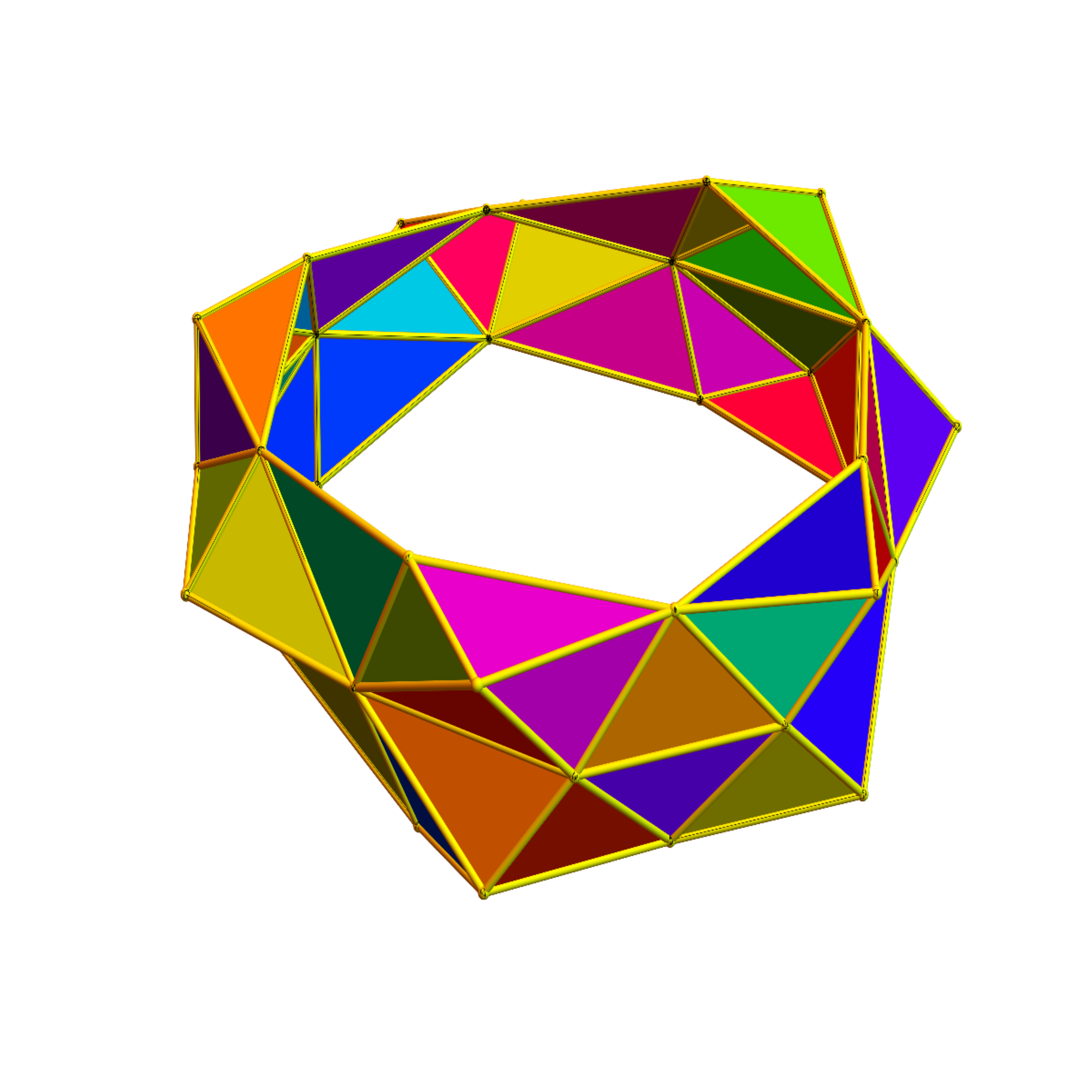}}
\caption{
\label{Cylinder}
A cylinder which can was edge refined, even so both 
boundaries had length $8$. We see that the mod 3 property 
does no more hold here. 
}
\end{figure}

\begin{figure}
\scalebox{0.07}{\includegraphics{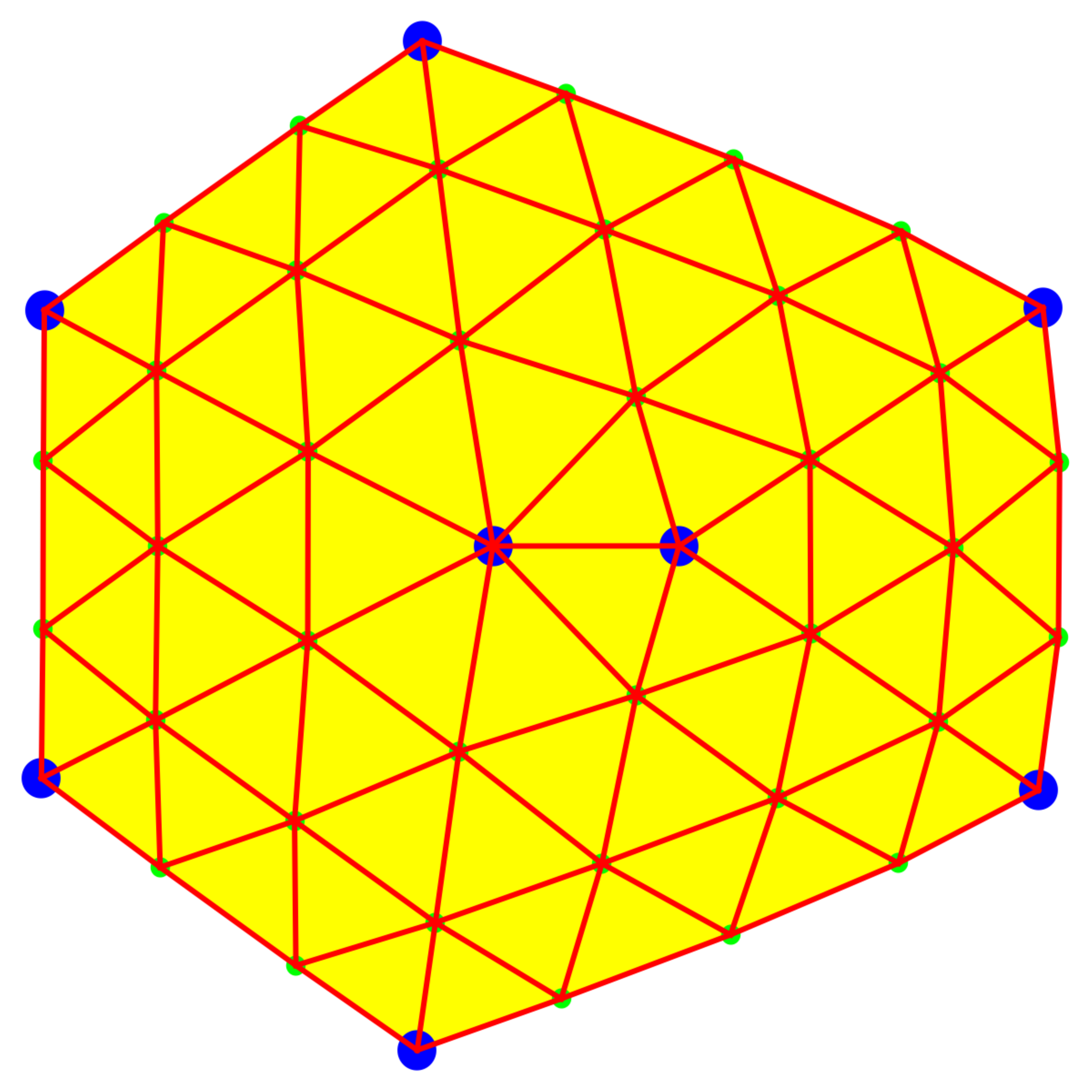}}
\scalebox{0.07}{\includegraphics{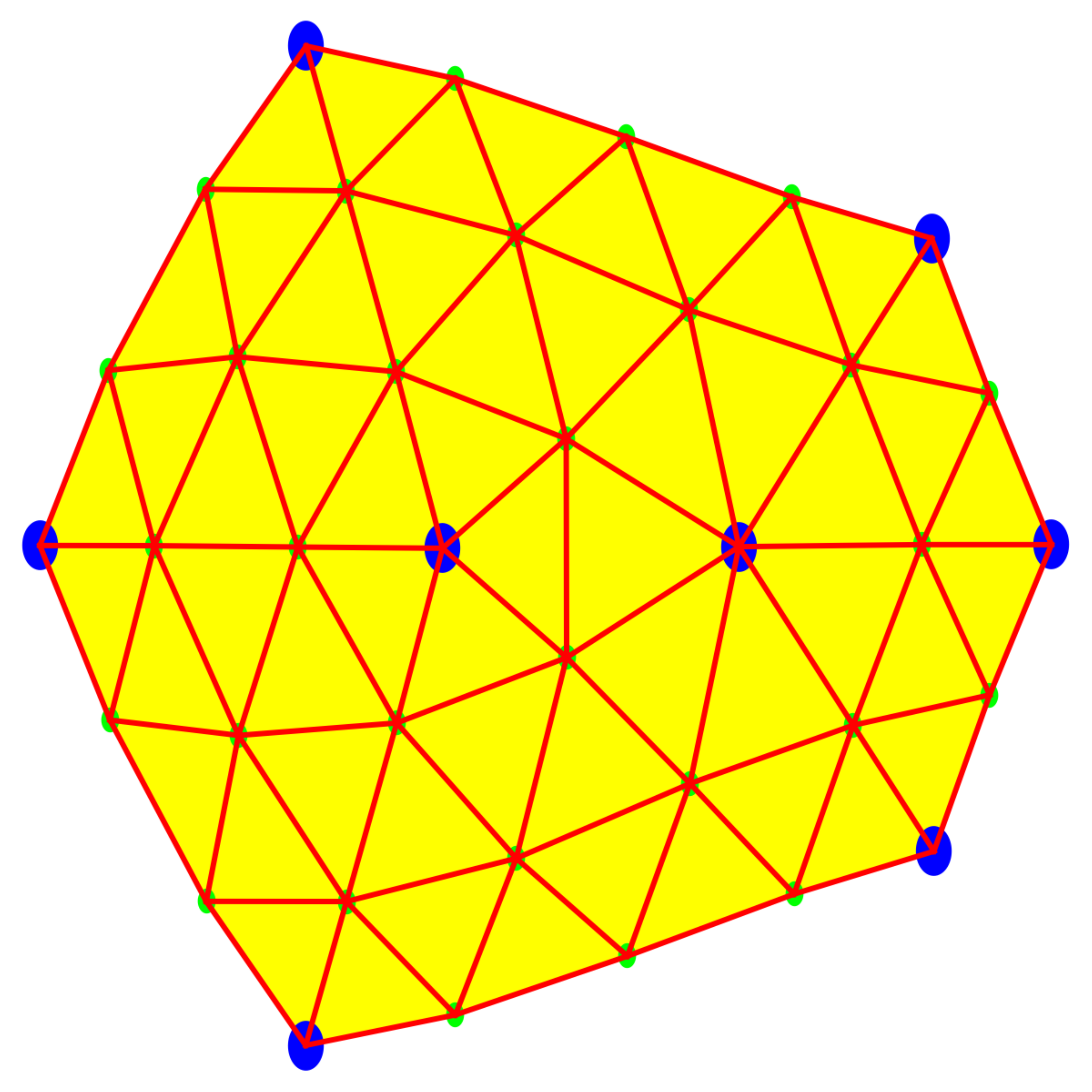}}
\caption{
\label{Refine}
We see two discs containing a 5-7 particle pair.
The situation to the left can not be edge refined within the
interior to become Eulerian. The second one can be
refined to become Eulerian because the boundary length is divisible by 3.
}
\end{figure}

\begin{figure}
\scalebox{0.62}{\includegraphics{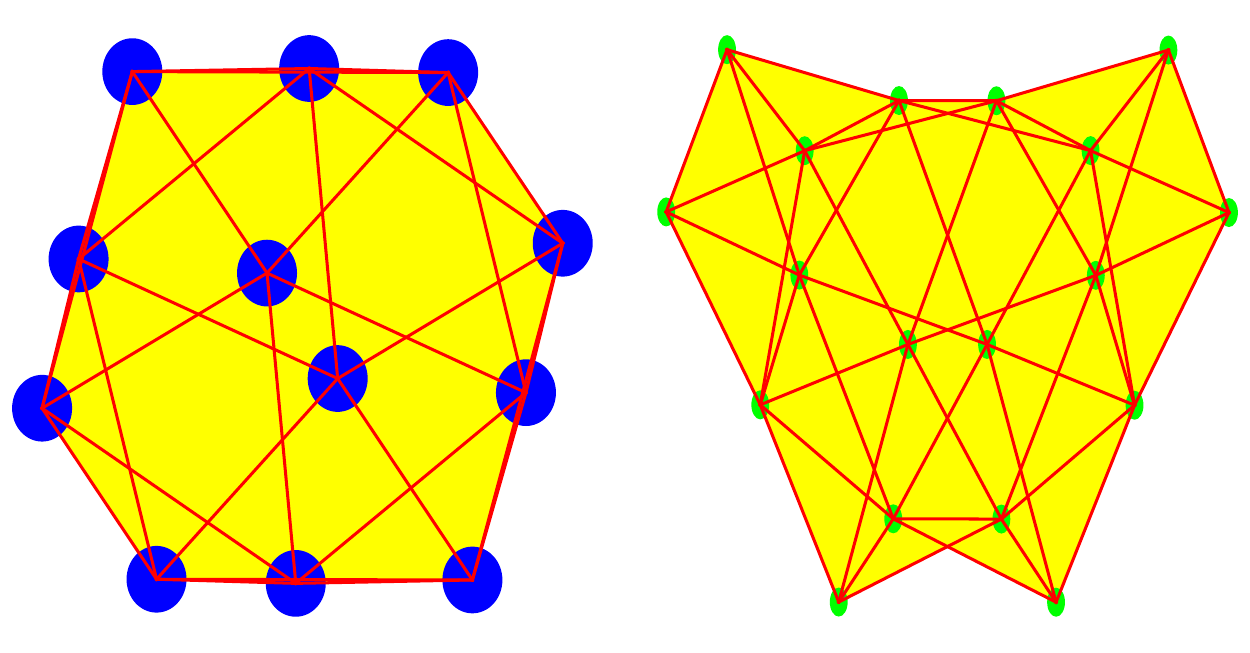}}
\scalebox{0.62}{\includegraphics{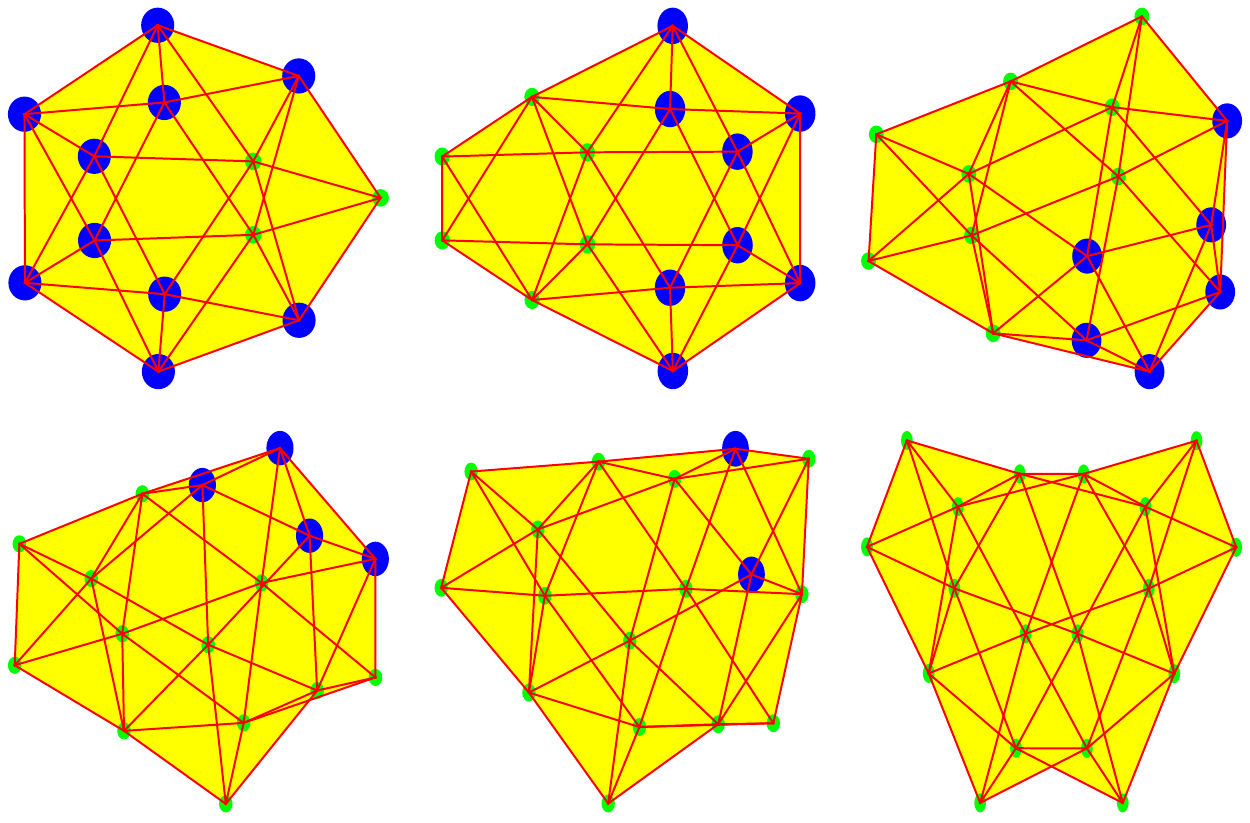}}
\caption{
\label{Refine}
Refining a s-sphere, an icosahedron graph. Initially, all vertices have
odd degree. After edge refinement, all degrees are even. 
}
\end{figure}

\begin{figure}
\scalebox{0.62}{\includegraphics{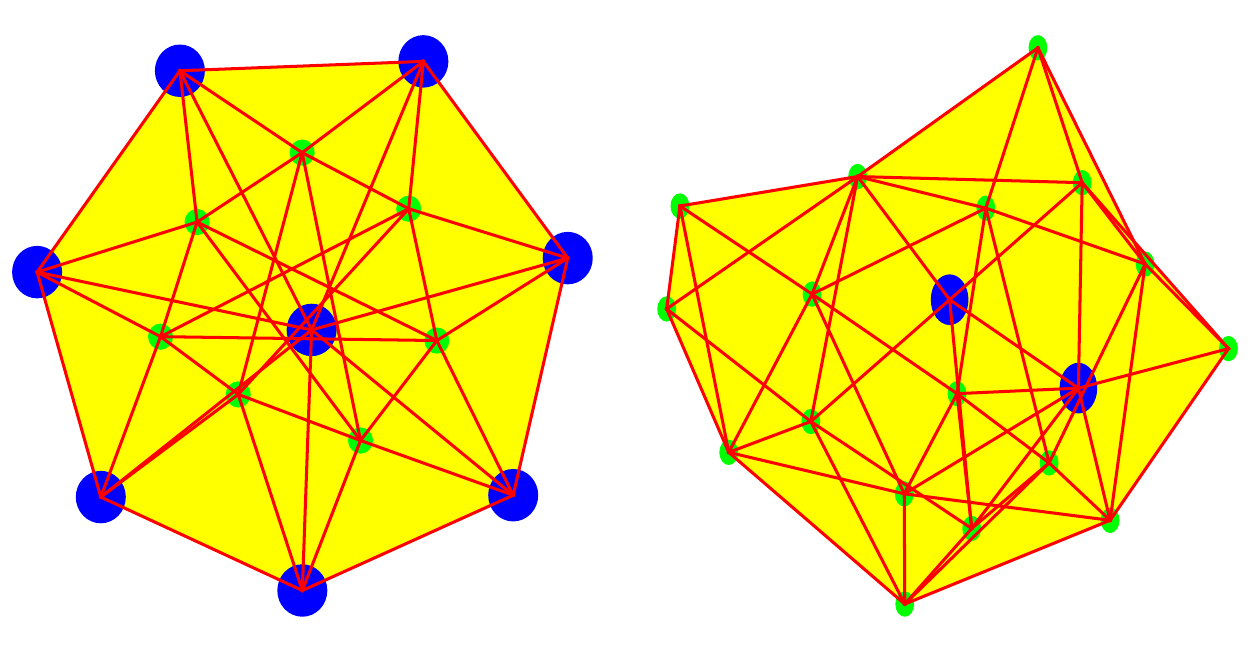}}
\scalebox{0.62}{\includegraphics{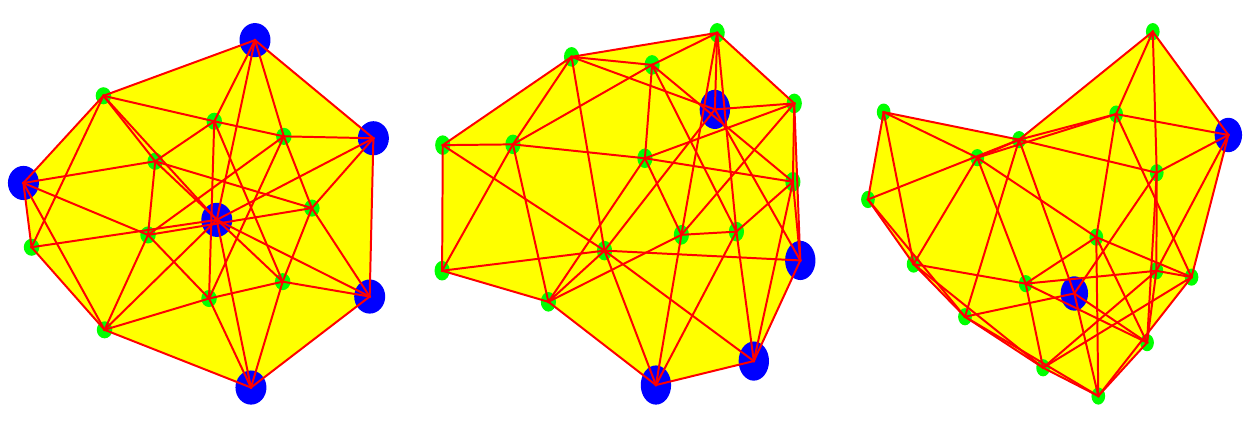}}
\caption{
\label{Refine}
Refining a small projective plane. In this example, we can not
get rid of the two last particles at first directly.
}
\end{figure}

\begin{figure}
\scalebox{0.62}{\includegraphics{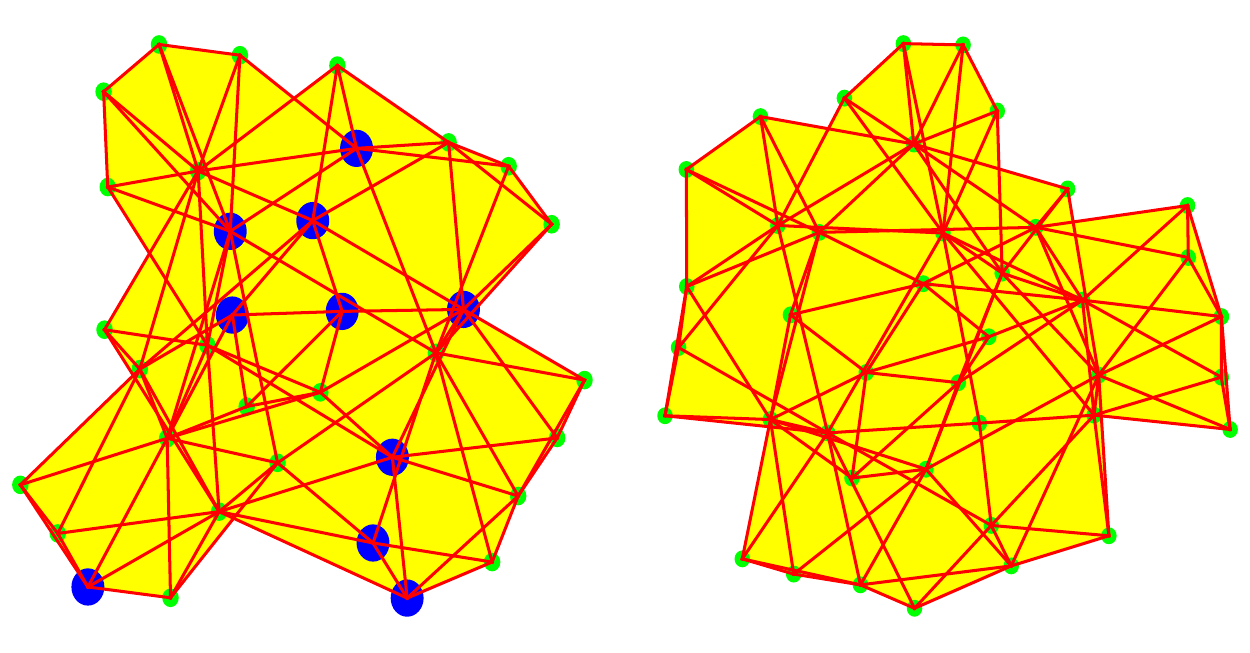}}
\scalebox{0.62}{\includegraphics{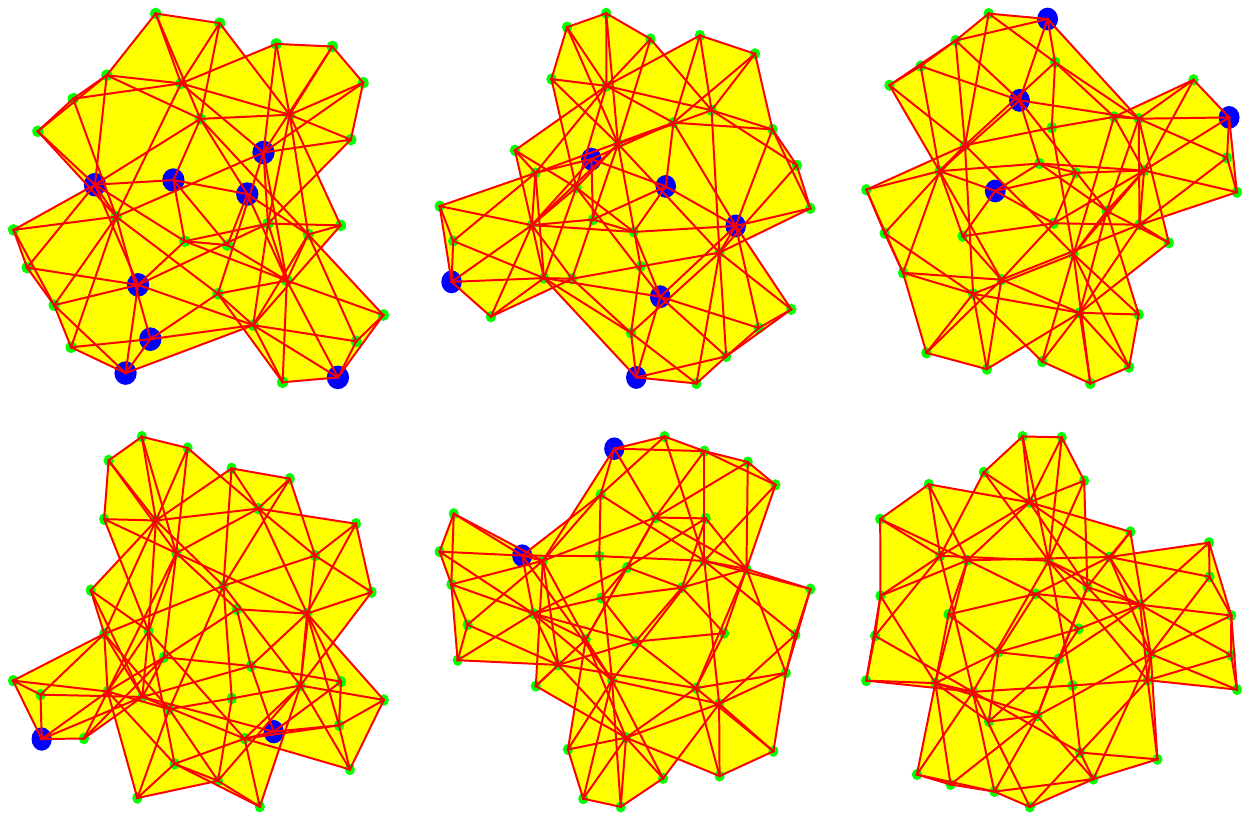}}
\caption{
\label{Refine}
Refining a small projective plane. We took the same example but first
made 20 random refinements to allow for more room. Now, we could get
rid of all the particles. 
}
\end{figure}

\begin{figure}
\scalebox{0.62}{\includegraphics{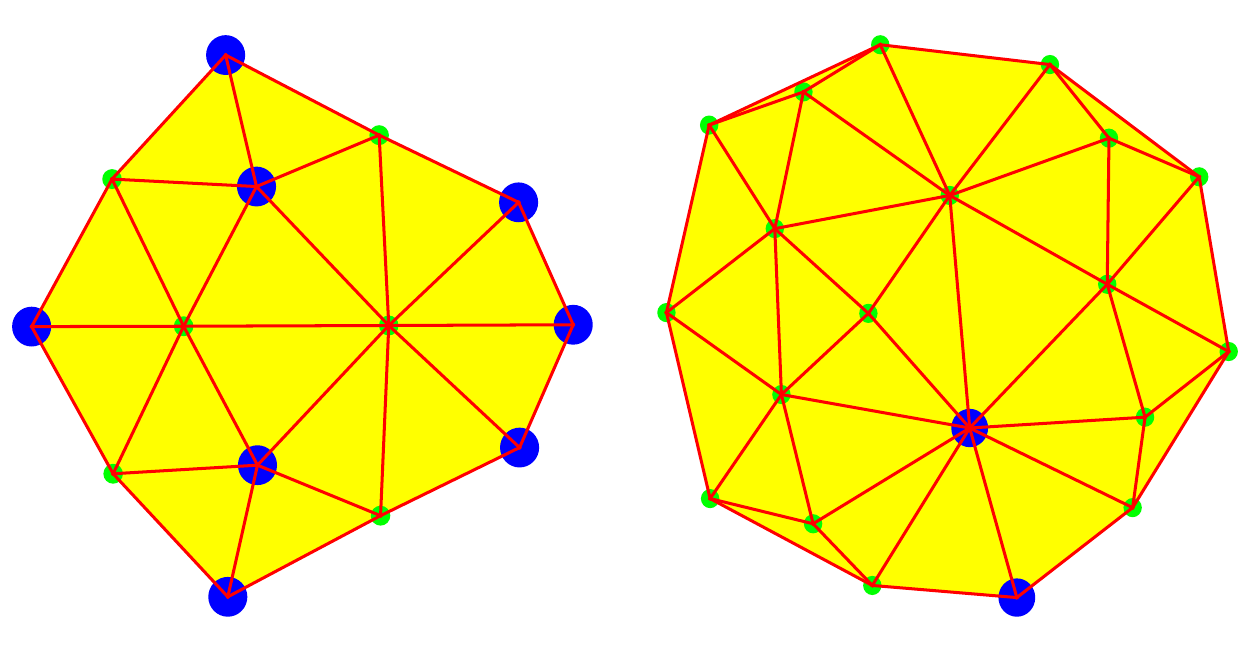}}
\scalebox{0.62}{\includegraphics{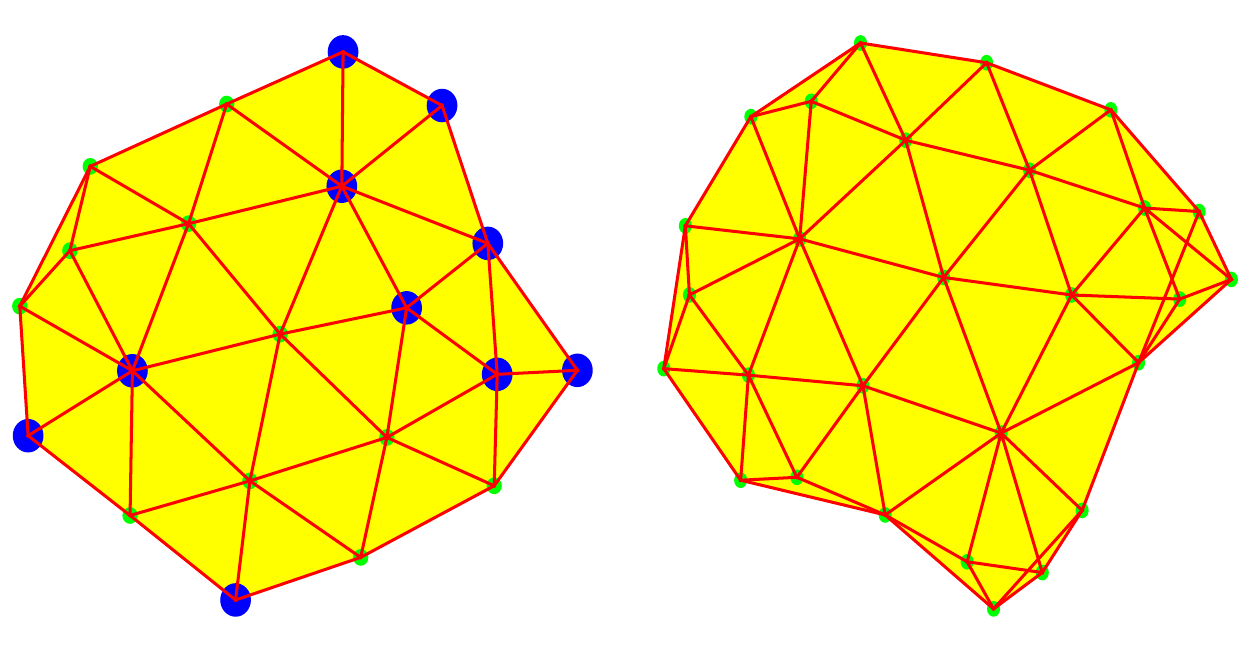}}
\caption{
\label{Refine}
Above: For a 2-ball with boundary of length 11, we can not get rid of all
odd vertices. 
Below: For a 2-ball with boundary of length 12, we can get rid of all
odd vertices.
}
\end{figure}

\begin{figure}
\scalebox{0.12}{\includegraphics{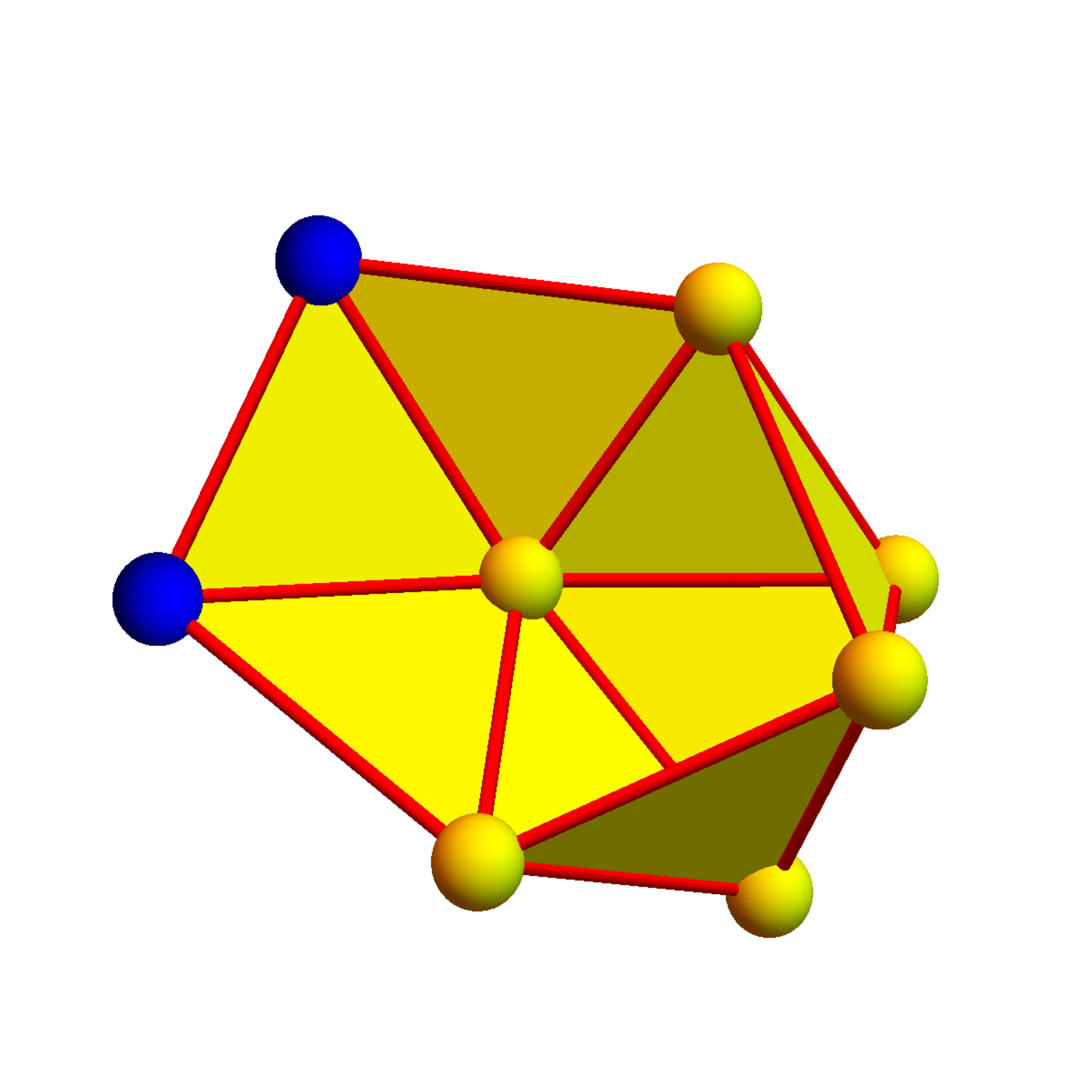}}
\caption{
\label{wheel5}
We see a disk of length 5 which can not be refined. There 
are two odd degree vertices present. Taking away the triangle
containing the two vertices (or adding an other triangle)
produces a graph with boundary
length $6$ which is Eulerian. We see from this picture that
in a disk with length divisible by 3, it is not possible to 
have two adjacent isolated odd degree vertices. 
}
\end{figure}

\begin{figure}
\scalebox{0.08}{\includegraphics{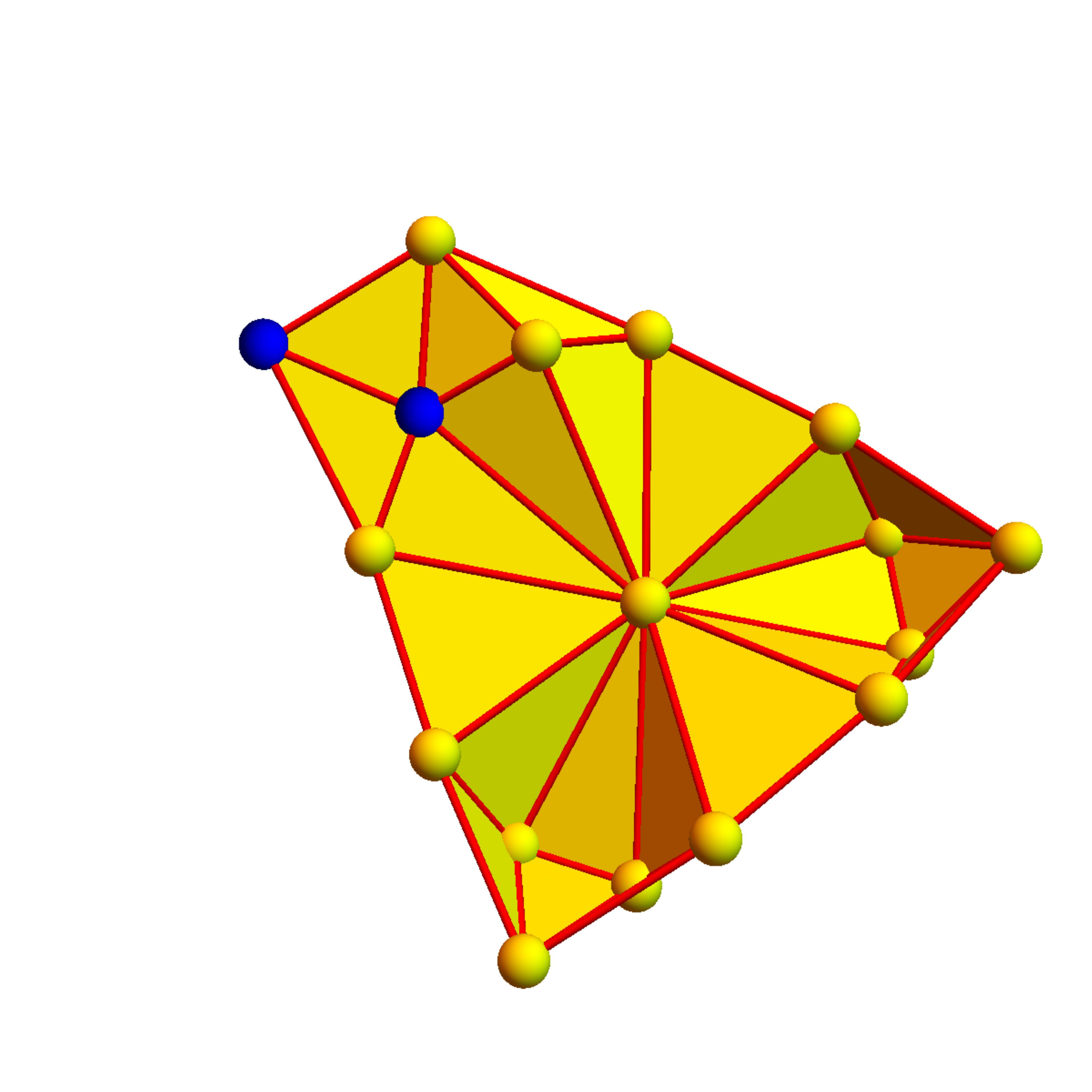}}
\scalebox{0.08}{\includegraphics{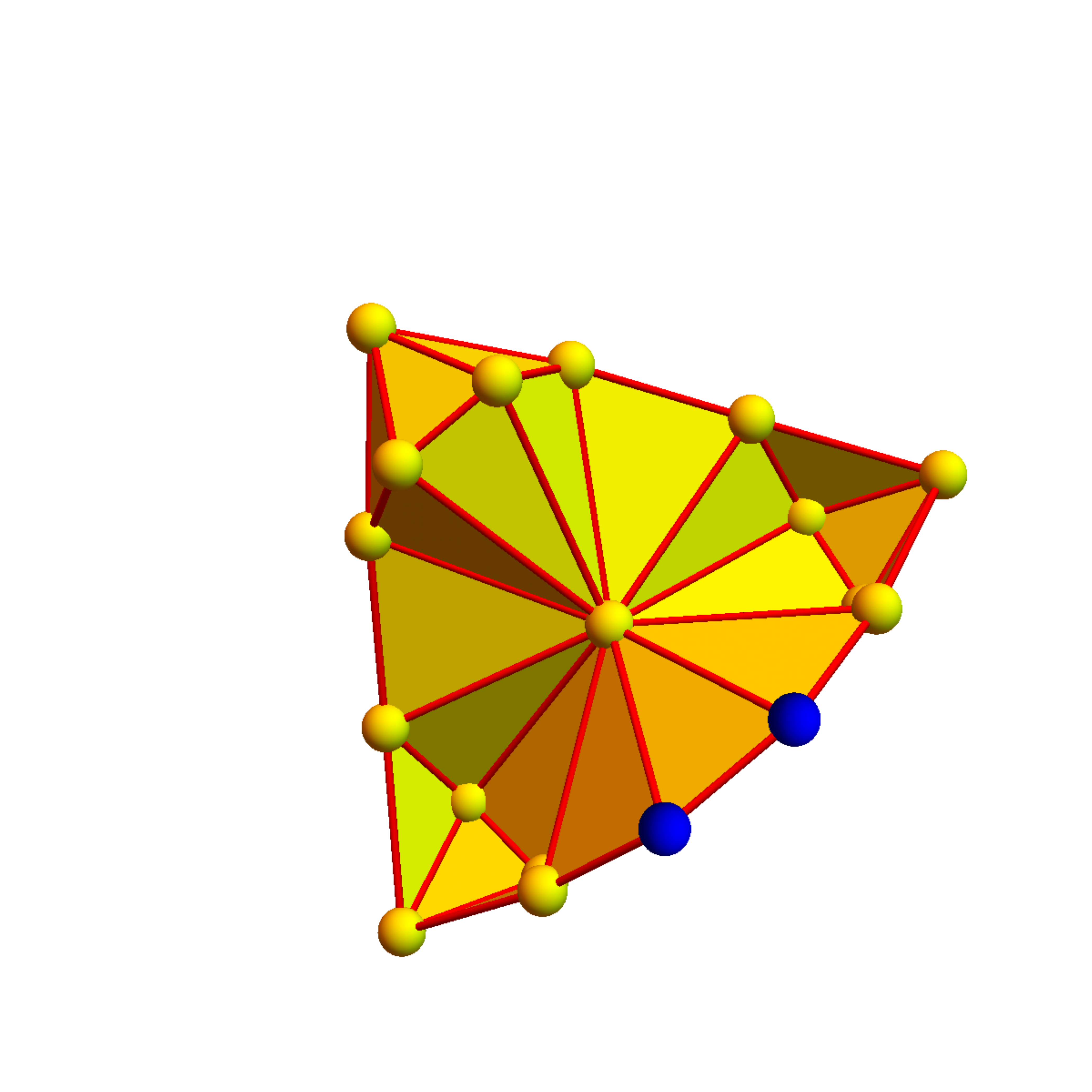}}
\scalebox{0.08}{\includegraphics{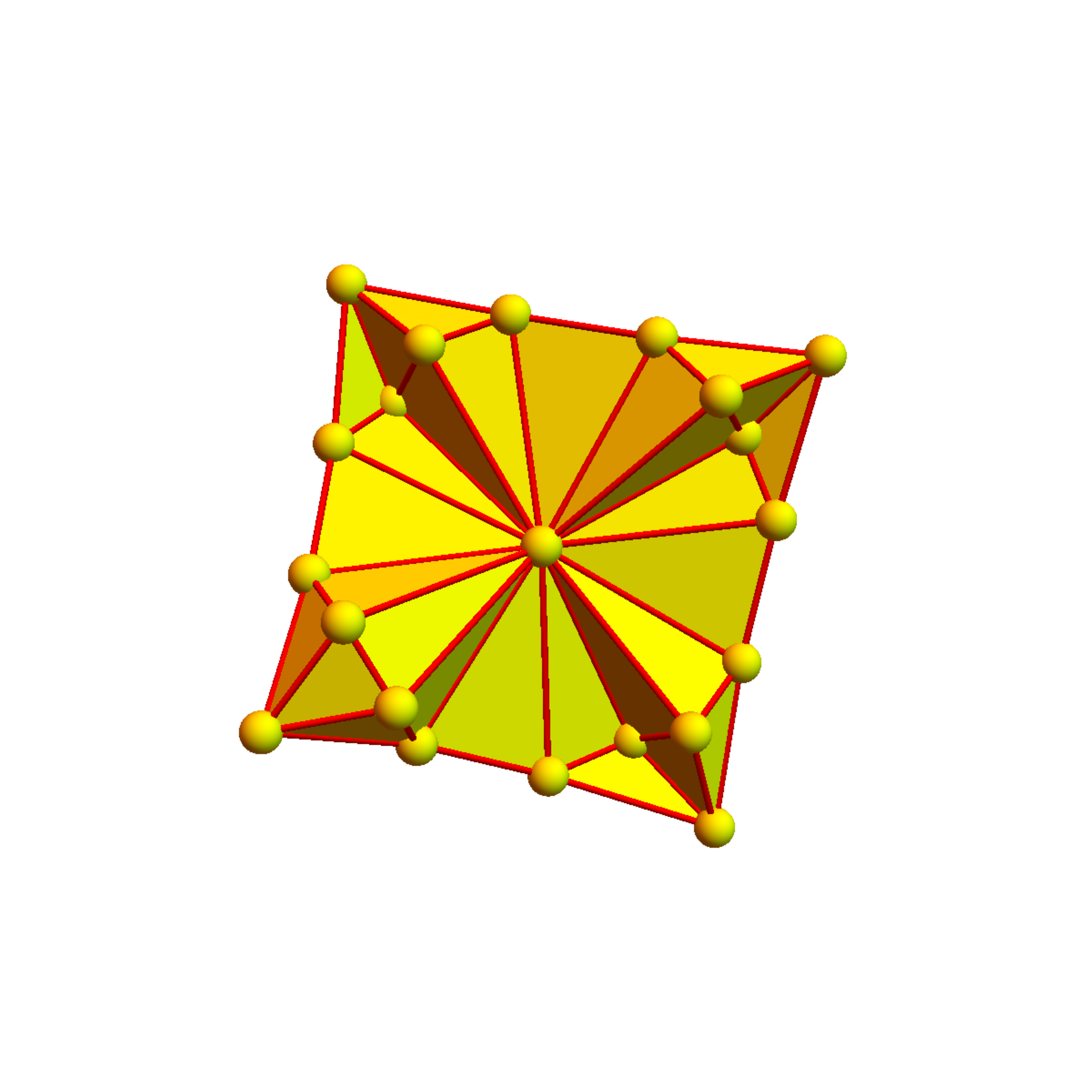}}
\caption{
\label{Refine}
Refined wheel graphs with boundary length $10,11,12$.
In the first two cases, we can not get rid of all 
odd degree vertices. 
}
\end{figure}

\begin{figure}
\scalebox{0.18}{\includegraphics{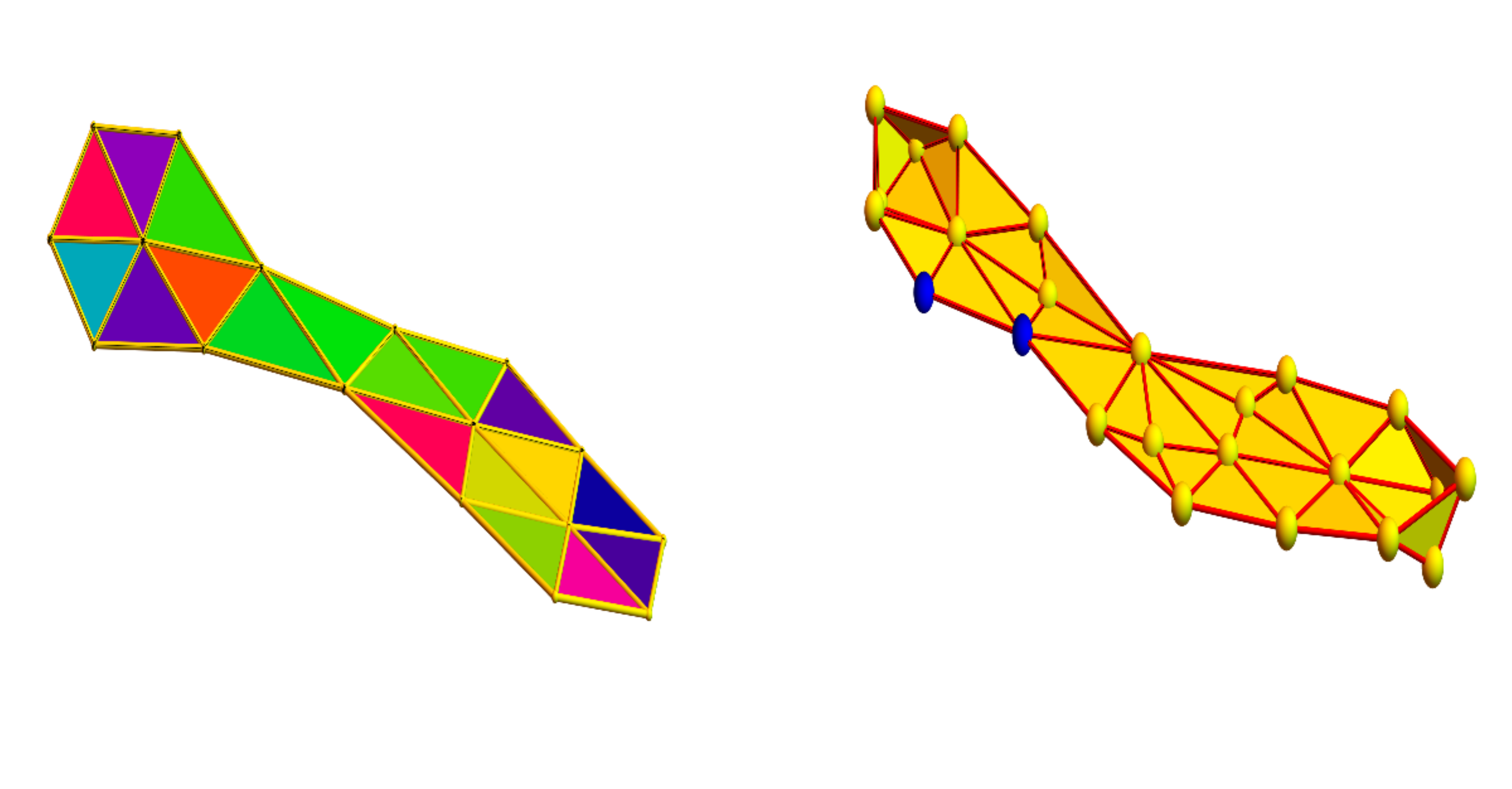}}
\scalebox{0.18}{\includegraphics{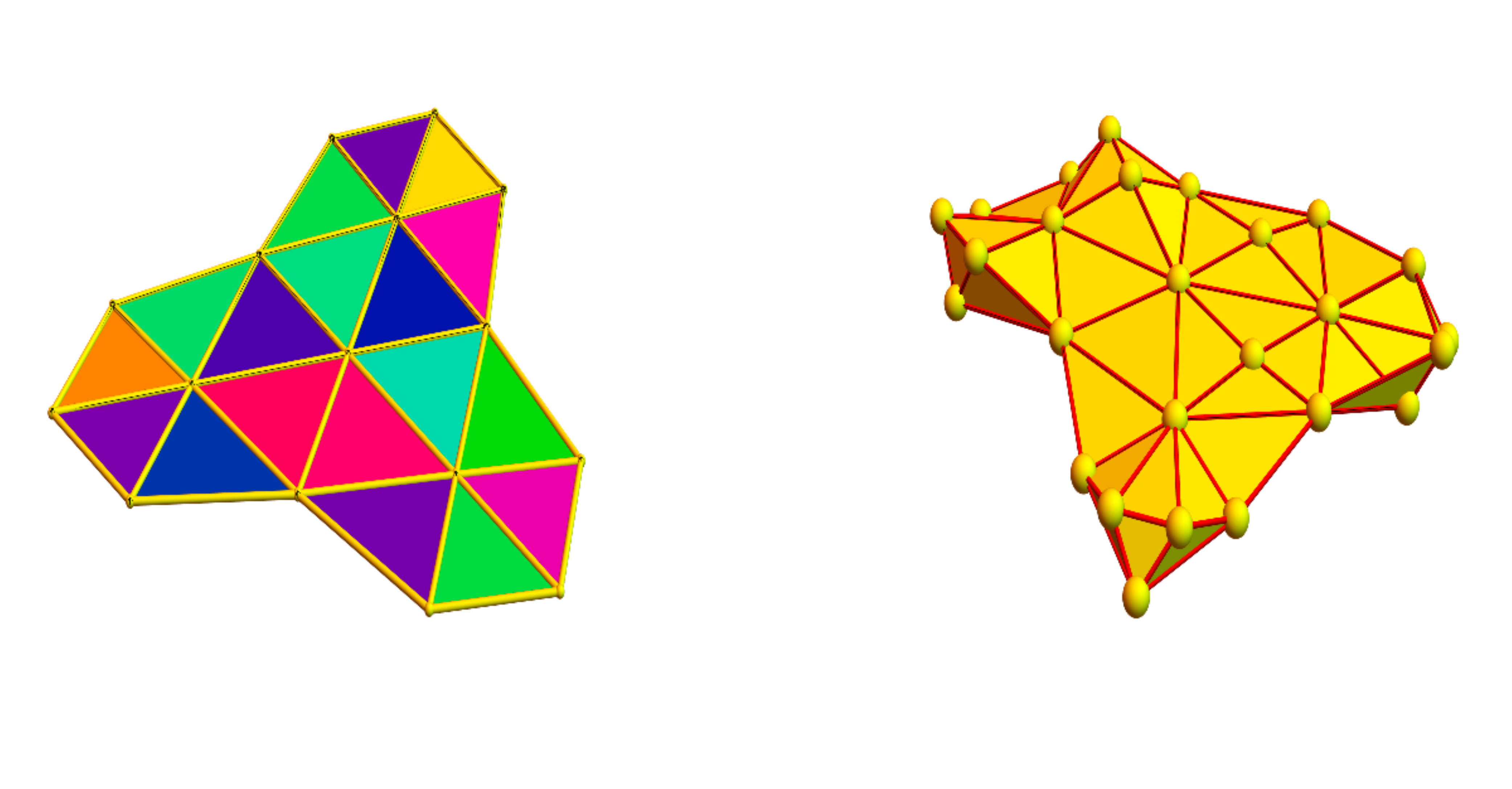}}
\caption{
\label{Thin}
After edging away the boundary, a disc can become thin and
fall apart. Here we see two thin disks and their edge
refinement. In the first one the boundary length was 14, 
in the second, it was 12. The second one could be refined to 
become Eulerian, the second one not. The billiards in the 
second case is far from ergodic. There are 7 ergodic components.
}
\end{figure}

\section*{4-color Strategy}

\paragraph{}
Let us sketch a strategy to prove the sphere conjecture in the case $d=2$, where
it is a known theorem, the 4-color theorem. While not fleshed out yet, it motivates 
and uses both theorems proven here. In the following, we assume that $G$ is a fixed $d$-sphere.  \\

{\bf Step 1:} There exists a ball $B$ with boundary $G$ and central vertex $v$
for which $S(v)$ is contained in the interior of $B$. Proof. First start with a cone
extension $B_1=x+G$, then do edge refinements until no edge containing $v$ can reach
$G$ any more. The unit sphere $S(x)$ is now completely in the interior of the ball $B$. \\

{\bf Step 2:} Use edge refinements so that for every interior point, every unit sphere
$S(v)$ intersects $G$ only in a $1,2,\cdots (d-1)$-simplex (or the empty graph). \\

{\bf Step 3:} Use edge refinements with edges in $S(x)$ to render the unit sphere 
$S(v)$ Eulerian. This is possible by Theorem I. Now, the edge degrees of all edges containing $x$
have even degree. The edges containing $x$ will no more be touched from now on. 
The already clean out part of $H$ is called $C$. 

{\bf Step 4:} We additionally can assure that there is a
vertex $v$ for which the vertex degree is divisible by $3$. 
The reason is that if we cut along a closed geodesic, 
then the vertex degree of each vertex changes by $2$.  \\

{\bf Step 5:} Take this vertex $y \in S(x)$ and look at the sphere $U=S(y)$. 
It contains a $(d-1)$-ball $V = B(x) \cap U$ and $W=B(x) \cap G$. 
render $U$ Eulerian, away from $V \cup W$ without cutting edges in $V$.
This means that all edges in the interior
of $B(y)$ have even degree. From now on, we will no more touch any edges containing
$y$. Call $C=B(x) \cup B(y)$. This is a ball. We will no more touch this 
cleaned out set $C$ except refinements which assure that the boundary of the 
disks to be cleaned out has length divisible by 3.  \\

{\bf Step 6:} Take a vertex $z \in \delta C$ and look at the sphere $U=S(z)$.
The graph $B= S(z) \cap C$ is a ball in $S(z)$. Also look at $K=S(z) \cap G$
which is a simplex. We can render $U$ Eulerian away from
$V \cup W$. Again rename $C=B(x) \cup B(y) \cup B(z)$ which is the new clean set.  \\

{\bf Step 7:} Continue like that. In each case, we have a ball $C$ with boundary $\delta C$.
We make sure to chose a new point $z$ such that $C \cap B(z)$ remains a ball and 
perform additional refinements if the length of the boundary is not divisible by 3.  \\

{\bf Step 8:} Once there are no vertices in $S(x)$ left, all the interior degrees are even
and the graph $H$ is Eulerian. It can be colored by 4 colors. This colors the boundary $G$. 

\section{Code}

\paragraph{}
Here is Mathematica code which computes the ergodic components of a geodesic or for a billiard.
The code assumes that feed in a graph for which all vertex degrees are even. As a test we
try out the Octahedron graph, which has 3 ergodic components. You can copy paste the code from 
the ArXiv version of this paper. 

\begin{tiny} \lstset{language=Mathematica} \lstset{frameround=fttt} 
\begin{lstlisting}[frame=single]
UnitSphere[s_,a_]:=VertexDelete[NeighborhoodGraph[s,a],a];
ErgodicComponents[s_]:=Module[{EE=EdgeList[s],e0,e,E1,EEE={}},
CircleQ[t_]:=2 Max[GraphDistanceMatrix[t]]==Floor[Length[VertexList[t]]];
Antipodal[t_,x_]:= Module[{dist=GraphDistance[t, x], max, k, vv},
  vv = VertexList[t]; max = Max[dist]; k=Flatten[Position[dist, max]];
  H=Table[vv[[k[[j]]]], {j,Length[k]}];H];
Billiard[t_,x_]:=Module[{U=FindHamiltonianPath[t],V=VertexList[t]},
  {U[[(1+Length[U]-Position[U,x][[1,1]])]]}];
GeodesicStep[xy_]:=Module[{t,z}, t = UnitSphere[s,xy[[2]]];
  If[CircleQ[t],z = Antipodal[t, xy[[1]]], z=Billiard[t,xy[[1]]]]; 
  {xy[[2]],First[z]}];
EE=Table[Sort[{EE[[k,1]],EE[[k,2]]}],{k,Length[EE]}];
While[Length[EE]>0,e0=First[EE];e0={e0[[1]],e0[[2]]};e=GeodesicStep[e0];
  E1={e0,e}; While[Not[e==e0],e=GeodesicStep[e]; E1=Append[E1,Sort[e]]];
  EE=Complement[EE,E1]; EEE=Append[EEE,E1]; ]; EEE ];
OctahedronGraph=UndirectedGraph[Graph[
{1->2,1->3,1->4,1->5,2->3,2->4,2->6,3->5,3->6,4->5,4->6,5->6}]];
ErgodicComponents[OctahedronGraph]
\end{lstlisting} 
\end{tiny}

\paragraph{}
And here is an example of an ergodic billiard. Two ergodic components are reported, but one
is the table, the boundary of the graph. 

\begin{tiny} 
\lstset{language=Mathematica} \lstset{frameround=fttt}
\begin{lstlisting}[frame=single]
Bunimovich=UndirectedGraph[Graph[
{1->2,1->3,1->4,1->5,2->4,2->6,2->7,3->5,3->8,3->9,3->10,
3->11,4->5,4->7,4->8,4->10,4->12,4->13,4->14,4->15,5->8,
6->7,6->12,6->15,7->15,8->10,9->11,9->16,9->17,10->11,
10->13,10->16,10->17,10->18,11->17,12->14,12->15,13->14,
13->18,14->18,16->17,16->18}]];
ErgodicComponents[Bunimovich]
\end{lstlisting} 
\end{tiny}

\paragraph{}
Here is an example of an ergodic 2 torus. There is only one ergodic component. 

\begin{tiny} 
\lstset{language=Mathematica} \lstset{frameround=fttt}
\begin{lstlisting}[frame=single]
ErgodicTorus=UndirectedGraph[Graph[
{1->2,1->6,1->13,1->16,1->34,1->62,1->88,1->120,
2->6,2->13,2->14,2->26,2->66,2->110,2->134,6->34,
6->66,6->32,6->38,6->56,6->78,6->96,6->118,6->130,
6->156,13->16,13->134,13->28,13->9,13->82,13->42,
13->50,13->136,16->62,16->15,16->12,16->50,16->146,
16->90,34->120,34->5,34->10,34->130,62->88,62->52,
62->122,62->90,88->120,88->4,88->5,88->122,120->5,
14->26,14->110,14->3,14->15,14->80,14->10,14->132,
14->112,14->148,26->134,26->9,26->82,26->132,
66->110,66->3,66->80,66->126,66->78,66->158,110->80,
134->82,3->15,3->54,3->64,3->80,3->124,3->126,3->142,
3->160,15->124,15->52,15->11,15->92,15->138,15->18,
15->22,15->44,15->146,15->148,15->30,15->90,15->140,
54->162,54->142,54->160,54->8,54->46, 54->7,54->144,
64->124,64->142,64->46,64->52,64->140,124->140,126->160,
126->7,126->108,126->158,142->46, 160->7,4->162,4->5,
4->8,4->20,4->46,4->52,4->122,5->20,5->10,5->72,5->102,
8->162,8->20,8->12,8->144,8->24,8->28,8->40,8->48,
20->72,20->40,20->9,20->84,20->100,46->162,46->52,
52->122,52->90,52->140,10->102,10->96,10->118,10->130,
10->9,10->98,10->132,10->36,10->58,10->60,10->68,
10->74,10->86,10->112,10->150,72->102,72->60,72->86,
72->100,102->86,32->38,32->56,32->70,32->94,32->106,
38->96,38->58,38->94,38->116,56->156,56->7,56->11,
56->108,56->154,56->106,56->76,56->128,78->156,78->108,
78->128,78->158,96->118,96->68,96->116,118->130,
156->128,7->11,7->12,7->108,7->144,11->12,11->70,
11->92,11->138,11->154,12->144,12->24,12->138,12->18,
12->22,12->42,12->44,12->50,12->136,12->146,108->128,
108->158,24->48,24->42,28->40,28->48,28->9,28->42,
40->9,48->42,9->82,9->84,9->98,9->132,84->98,84->36,
84->60,84->100,98->36,36->60,58->68,58->74,58->70,
58->94,58->116,58->152,60->86,60->100,68->116,74->150,
74->104,74->152,74->114,112->150,112->148,112->30,
112->104,150->104,70->92,70->154,70->30,70->104,
70->94,70->106,70->76,70->152,70->114,92->30,138->44,
154->76,18->22,18->44,22->146,42->136,50->136,148->30,
30->104,104->114,106->76,152->114}]];
ErgodicComponents[ErgodicTorus]
\end{lstlisting} \end{tiny}

\paragraph{}
And here is the code which refines a graph without boundary 
so that it becomes  Eulerian. The procedure is done by geodesic cutting.
There is some randomness built in: after choosing an odd degree
vertex, we chose the direction randomly. We did that originally to 
test whether the number of ergodic components depends on the cutting. 
The code shows that it does and it also allows to try again and again
until an ergodic one is reached. 

\begin{tiny} \lstset{language=Mathematica} 
\lstset{frameround=fttt} \begin{lstlisting}[frame=single]
OddVertices[s_] := Module[{V={},v=VertexList[s],d=VertexDegree[s]},
   Do[If[OddQ[d[[k]]],V=Append[V,v[[k]]]],{k,Length[v]}]; V];
UnitSphere[s_,a_]:=VertexDelete[NeighborhoodGraph[s,a],a];
Antipodal[t_, x_] := Module[{dist = GraphDistance[t, x], max,k,vv},
   vv = VertexList[t]; max=Max[dist];k=Flatten[Position[dist, max]];
   Table[vv[[k[[j]]]],{j,Length[k]}]];
GraphSubdivide[s_,{a_,b_}]:=Module[{e1,v,e,vv,i,j,n,pp,t1,t2,sss,c,ee},
    v=VertexList[s]; e=EdgeRules[s]; n = Max[v]+1; vv=Append[v,n+1];
    t1=UnitSphere[s,a];   t2=UnitSphere[s,b];
    c=Intersection[VertexList[t1],VertexList[t2]];
    ee=Union[e,{(n+1)->a,(n+1)->b},Table[(n+1)->c[[k]],{k,Length[c]}]];
    ee=Complement[ee,{a->b,b->a}]; UndirectedGraph[Graph[ee]]]
MakeEulerian[s_]:=Module[{t=s}, SelfHeal:=Module[{},
   T[{x_,y_}]:= Module[{z,h,v,w,xx,yy},h=UnitSphere[t,y];
   z = Antipodal[h, x]; v = VertexList[t];
   If[Length[z]==1, xx=y; yy=z[[1]]; w=v,
    xx=y; t=GraphSubdivide[t,z]; w=v; v=VertexList[t];
    yy=First[Complement[v,w]]]; {xx,yy}]; vv=OddVertices[t];
   If[Length[vv]>0, y0=RandomChoice[vv];
     x0 = RandomChoice[VertexList[UnitSphere[t,y0]]];
     {x, y} = T[{x0,y0}]; X = {{x0,y0}, {x,y}};
     While[Not[MemberQ[vv,y]],{x,y}=T[{x, y}]; X=Append[X,{x,y}]]]];
   While[Length[OddVertices[t]] > 0,SelfHeal];
t];
IcosaGraph=UndirectedGraph[Graph[
{1->2,1->3,1->4,1->5,1->6,2->5,2->6,2->9,2->10,3->4,3->5,3->8,
3->11,4->6,4->8,4->12,5->9,5->11,6->10,6->12,7->8,7->9,7->10,
7->11,7->12,8->11,8->12,9->10,9->11,10->12}]];
Do[Print[Length[ErgodicComponents[MakeEulerian[IcosaGraph]]]],{9}]
\end{lstlisting} \end{tiny}

\bibliographystyle{plain}

\begin{thebibliography}{10}

\bibitem{ChernovMarkarian}
N.~Chernof and R.~Markarian.
\newblock {\em Chaotic billiards}.
\newblock AMS, 2006.

\bibitem{Crowe1969}
D.W. Crowe.
\newblock Nearly regular polyhedra with two exceptional faces.
\newblock In G.~Chartrand and S.F. Kapoor, editors, {\em The Many Facets of
  Graph Theory}, pages 63--76. Springer Berlin Heidelberg, 1969.

\bibitem{Eberhard1891}
E.~Eberhard.
\newblock {\em Morphologie der Polyeder}.
\newblock Teubner Verlag, 1891.

\bibitem{Fisk1978}
S.~Fisk.
\newblock The nonexistence of colorings.
\newblock {\em Journal of Combinatorial Theory B}, 24:247--2480, 1978.

\bibitem{gruenbaum}
B.~Gr\"unbaum.
\newblock {\em Convex Polytopes}.
\newblock Springer, 2003.

\bibitem{Izmestiev2013}
I.~Izmestiev.
\newblock Courbure discr\`ete: th\'eoryie et applications.
\newblock {\em CIRM}, 3:151--157, 2013.

\bibitem{Izmestiev2015}
I.~Izmestiev.
\newblock Color or cover.
\newblock https://arxiv.org/pdf/1503.00605.pdf, 2015.

\bibitem{IKRSS}
I.~Izmestiev, R.B.Kusner, G.Rote, B.~Springborn, and J.M. Sullivan.
\newblock There is no triangulation of the torus with vertex degrees
  5,6,....,6,7 and related results: geometric proofs for combinatorial
  theorems.
\newblock {\em Geom. Dedicata}, 166:15--29, 2013.

\bibitem{Jendrol1975}
S.~Jendrol.
\newblock On the non-existence of certain nearly regular planar maps with two
  exceptional faces.
\newblock {\em Mat Cas}, 25:159--164, 1975.

\bibitem{JendrolJucovic72}
S.~Jendrol and E.~Jucovic.
\newblock On a conjecture of b. {Gr\"unbaum}.
\newblock {\em Discrete Math}, 10, 1972.

\bibitem{knillgraphcoloring}
O.~Knill.
\newblock Coloring graphs using topology.
\newblock {{\\}http://arxiv.org/abs/1410.3173}, 2014.

\bibitem{KnillTopology}
O.~Knill.
\newblock A notion of graph homeomorphism.
\newblock {{\\}http://arxiv.org/abs/1401.2819}, 2014.

\bibitem{KnillEulerian}
O.~Knill.
\newblock Graphs with {E}ulerian unit spheres.
\newblock {\\}http://arxiv.org/abs/1501.03116, 2015.

\bibitem{Tab95}
S.~Tabachnikov.
\newblock {\em Billiards}.
\newblock Panoramas et synth\`eses. Soci\'et\'e Math\'ematique de France, 1995.

\end{thebibliography}

\end{document}